\documentclass[a4paper]{article}

\usepackage{microtype}

\usepackage{geometry}
\usepackage{ifthen}
\newif\iflong
\longtrue

\usepackage[UKenglish]{babel}
\usepackage[utf8]{inputenc}
\usepackage{amsmath,amssymb}

\usepackage[standard,amsmath,thmmarks]{ntheorem}

\makeatletter
\renewtheoremstyle{plain}%
{\item[\hskip\labelsep \theorem@headerfont ##1\ ##2 \theorem@separator]}%
{\item[\hskip\labelsep \theorem@headerfont ##1\ ##2\ \normalfont({\sffamily##3})  \theorem@separator]}
\makeatother

\newtheorem{problem}[theorem]{Problem}
\renewtheorem{definition}[theorem]{Definition}
\renewtheorem{lemma}[theorem]{Lemma}
\renewtheorem{proposition}[theorem]{Proposition}
\renewtheorem{corollary}[theorem]{Corollary}
\renewtheorem{example}[theorem]{Example}
\renewtheorem{theorem}{Theorem}

\usepackage{microtype}

\title{Enumeration Complexity of\\ Poor Man's Propositional Dependence Logic\footnote{This work was partially supported by DFG project ME4279/1-2.}}

\author{Arne Meier \and Christian Reinbold}

\date{\small Leibniz Universität Hannover,\\ Institut für Theoretische Informatik,\\ Appelstrasse 4, 30167 Hannover, Germany,\\ \texttt{$\{$meier, reinbold$\}$@thi.uni-hannover.de}}

\usepackage{arbeit}

\newtoggle{classneg}
\togglefalse{classneg}

\def\defeq			{\mathrel{\mathop:}=}
\def\eqdef			{=\mathrel{\mathop:}}
\def\defequiv		{\mathrel{\mathop:}\Leftrightarrow}
\def\defdefeq		{\mathrel{\mathop{{\mathop:}{\mathop:}}}=}
\def\pow			#1{{\mathcal{P}(#1)}}
\def\natural		{{\mathbb{N}}}
\def\ftwo			{\mathbb{F}_2}
\def\restrict		#1{{\big |_{#1}}}
\def\setassign		#1{{2^{#1}}}
\def\setteam		#1{{\mathcal{P}\left(\setassign{#1}\right)}}
\def\setsatteams	#1{\mathcal{T}_{#1}}
\def\cardsatteams	#1{t_{#1}}
\def\setsatteamsz	#1{\mathcal{T}_{#1}^0}
\def\cardsatteamsz	#1{t_{#1}^0}

\def\sol			{{\mathrm{Sol}}}
\def\first			{\mathrm{first}}
\def\last			{\mathrm{last}}
\def\red			{\mathrm{red}}

\def\incP			{\ensuremath{\mathsf{IncP}}}
\def\delayP			{\ensuremath{\mathsf{DelayP}}}
\def\delaySpaceP	{\ensuremath{\mathsf{DelaySpaceP}}}
\def\incFPT			{\ensuremath{\mathsf{IncFPT}}}
\def\delayFPT		{\ensuremath{\mathsf{DelayFPT}}}
\def\FPT		{\ensuremath{\mathsf{FPT}}}
\def\p					{\ensuremath{\textsf{P}}}
\def\np					{\ensuremath{\textsf{NP}}}
\def\fpt				{\ensuremath{\textsf{FPT}}}

\def\pl				{{\mathcal{P\!L}^-}}
\def\pdl			{{\mathcal{PD\!L}^-}}
\def\pdlnovee		{\pdl}

\def\var			{{\text{Var}}}

\def\classvee		{\varovee}
\def\dep			#1{{=\!\!(#1)}}
\def\acts			{{\ \circlearrowleft\ }}
\def\bigo			{{\mathcal{O}}}
\def\corresponds	{\hat{=}}

\def\enumteampl			{{\textsc{EnumTeam}}}
\def\penumteampl		{{\textsc{p-EnumTeam}}}
\def\enumteamplsize		{{\textsc{EnumTeamSize}}}
\def\penumteamplsize	{{\textsc{p-EnumTeamSize}}}

\SetKw{Not}{not}
\SetKw{Null}{null}
\SetKw{And}{and}
\SetKw{Or}{or}
\SetKw{New}{new}
\SetKw{True}{true}
\SetKw{Output}{output}
\SetKw{Continue}{continue}
\SetKw{Break}{break}
\SetKw{Accept}{accept}
\SetKw{Reject}{reject}
\SetKwFunction{First}{first}
\SetKwFunction{HasNext}{hasNext}
\SetKwFunction{Next}{next}
\SetKwFunction{Push}{push}
\SetKwFunction{Pop}{pop}
\SetKwComment{Comment}{/* }{ */}
\SetKwInput{KwDependencies}{Dependencies}
\SetKwData{Map}{Map}
\SetKwData{List}{List}
\SetKwData{Queue}{Queue}
\SetKwBlock{Parallel}{simultaneously}{end}

\begin{document}

\maketitle

\begin{abstract}
Dependence logics are a modern family of logics of independence and dependence which mimic notions of database theory.
In this paper, we aim to initiate the study of enumeration complexity in the field of dependence logics and thereby get a new point of view on enumerating answers of database queries. 
Consequently, as a first step, we investigate the problem of enumerating all satisfying teams of formulas from a given fragment of propositional dependence logic. 
We distinguish between restricting the team size by arbitrary functions and the parametrised version where the parameter is the team size. 
We show that a polynomial delay can be reached for polynomials and otherwise in the parametrised setting we reach FPT delay. 
However, the constructed enumeration algorithm with polynomial delay requires exponential space. 
We show that an incremental polynomial delay algorithm exists which uses polynomial space only. 
Negatively, we show that for the general problem without restricting the team size, an enumeration algorithm running in polynomial space cannot exist.
\end{abstract}

\section{Introduction}
Consider the simple database scheme containing a single table \textsc{Smartphone} with attributes \textsc{Manufacturer} (\textsc{M}), \textsc{Serial Number} (\textsc{SN}), \textsc{Manufacture Date} (\textsc{MD}) and \textsc{Bluetooth Support} (\textsc{BS}), where \textsc{Manufacturer} and \textsc{Serial Number} form the primary key. Now we are interested in all possible answers of a database query on \textsc{Smartphone} selecting entities with bluetooth support. In terms of dependence logic, a database instance $T$ conforms with the primary key condition if and only if $T \models \dep{\{\textsc{M}, \textsc{SN}\},\{\textsc{MD}, \textsc{BS}\}}$. Taking the selection of the query into consideration, we obtain the formula
$\dep{\{\textsc{M}, \textsc{SN}\},\{\textsc{MD}, \textsc{BS}\}} \wedge \textsc{BS}$ for which we would like to enumerate satisfying database instances. Since team semantics is commonly used in the area of dependence logic, we model the database instance $T$ as a team, that is, a set of assignments, such that each assignment represents a row in \textsc{Smartphone}.

The task of enumerating all solutions of a given instance is relevant in several prominent areas, e.g., one is interested in all tuples satisfying a database query, DNA sequencings, or all answers of a web search.
In enumeration complexity one is interested in outputting all solutions of a given problem instance without duplicates.
Often, the algorithmic stream of solutions has to obey a specific order, in particular, on such order increasingly arranges solutions with respect to their cost.
In view of this, the enumeration task (with respect to this order) outputs the cheapest solutions first.
Of course, all these algorithms usually are not running in polynomial time as there often exist more than polynomially many solutions.
As a result, one classifies these deterministic algorithms with respect to their \emph{delay} \cite{JY88,S09,S10}.
Informally, the delay is the time which elapses between two output solutions and guarantees a continuous stream of output solutions.
For instance, the class $\delayP$ then encompasses problems for which algorithms with a polynomial delay (in the input length) exist.
Another class relevant to this study is $\incP$, incremental $\p$.
For this class the delay of outputting the $i$th solution of an instance is polynomial in the input size \emph{plus} the index $i$ of the solution. 
Consequently, problem instances exhibiting exponentially many solutions eventually possess an exponential delay whereas, in the beginning, the delay was polynomial.
Some natural problems in this class are known such as enumerating all minimal triangulations \cite{DBLP:conf/pods/CarmeliKK17} or some problems for matroids \cite{begkm05}.

A prominent approach to attacking computationally hard problems is the framework of parametrised complexity by Downey and Fellows \cite{DowneyFellows13,DowneyFellows99}.
Essentially, one searches for a \emph{parameter} $k$ of a given problem such that the problem can be solved in time $f(k)\cdot n^{\mathcal O(1)}$ instead of $n^{f(k)}$ where $n$ is the input length and $f$ is an arbitrary recursive function.
Assuming that the parameter is slowly growing or even constant, then the first kind of algorithms is seen relevant for practice.
In these cases, one says that the problem is \emph{fixed parameter tractable}, or short, in $\FPT$.
A simple example here is the propositional satisfiability problem with the parametrisation \emph{numbers of variables}.
For this problem, the straightforward brute-force algorithm already yields $\FPT$.
Recently, this framework has been adapted to the field of enumeration by Creignou et~al.\ \cite{cmmjv16,ckmmov15}.
There, the authors introduced the corresponding enumeration classes $\delayFPT$ and $\incFPT$ and provided some characterisations of these classes.

In 2007, Jouko Väänänen introduced \emph{dependence logic} (DL) \cite{vaananen07} as a novel variant of Hintikka's \emph{independence-friendly logic}. 
This logic builds on top of compositional team semantics which emerges from the work of Hodges \cite{Hodges97c}.
In this logic, the satisfaction of formulas is interpreted on \emph{sets of assignments}, i.e., teams, instead of a single assignment as in classical Tarski semantics.
Significantly, this semantics allows for interpreting reasoning in this logic in the view of databases.
Essentially, a team then is nothing different than a database: its domain of variables is the set of columns and its (team) members, that is, assignments, can be seen as rows in the table.
As a result, the aforementioned dependence atoms allow for expressing key properties in databases, e.g., functional dependencies.
On that account, many research from the area of databases coalesced with scientific results from logic, complexity theory and further other disciplines \cite{DBLP:journals/dagstuhl-reports/AbramskyKVV13}.
As a result, within the team semantics setting several different formalisms have been investigated that have counterparts in database theory: inclusion and exclusion dependencies \cite{Fagin:1981:NFR:319587.319592,CASANOVA198429,Casanova:1983:TSV:588058.588065,GALLIANI201268}, functional dependence (the dependence atom $\dep{P,Q}$) \cite{vaananen07}, and independence \cite{DBLP:journals/sLogica/GradelV13}.
Such operators will be the topic of future research connecting to the here presented investigations.

To bring the motivation full circle, the study of enumeration in dependence logic is the same as investigating the enumeration of answers of specific database queries described over some formulas in some logic.
For instance, consider a dependence logic formula that specifies some database related properties such as functional or exclusion dependencies of some attributes.
Now one is interested in the question whether this specification is meaningful in the sense that there exists a database which obeys these properties.
This problem can be seen as the satisfiability problem in dependence logic.
Further connections to database theory have been exemplified by Hannula et~al.\ \cite{DBLP:journals/corr/HannulaKV17}.
The study of database queries is a deeply studied problem and exists for several decades now.
Our aim for this paper is to initiate the research on enumeration (in databases) from the perspective of dependence logic.
This modern family of logics might give fresh insights into this settled problem and produce new enumeration techniques that will help at databases as well.
From a computational complexity perspective, DL is well understood: most of the possible operator fragments have been classified \cite{hkvv15ext,HKMV16}.
However, it turned out that model checking and satisfiability for propositional dependence logic $\mathcal{PD\!L}$ are already $\np$-complete \cite{EL12,lv13}.
As a result, tractable enumeration of solutions in the full logic is impossible (unless $\p$ and $\np$ coincide) and we focus on a fragment of $\mathcal{PD\!L}$ which we will call, for historical reasons, the Poor Man's fragment \cite{h01}.

In this paper, we investigate the problem of enumerating all satisfying teams of a given Poor Man's propositional dependence logic formula.
In particular, we distinguish between restricting the team size by arbitrary functions $f$ and the parametrised version where the parameter is the team size.
We show that $\delayP$ can be reached if $f$ is a polynomial in the input length and otherwise the parametrised approach leads to $\delayFPT$.
However, the constructed $\delayP$ enumeration algorithm requires exponential space.
If one desires to eliminate this unsatisfactory space requirement, we show that this can be achieved by paying the price of an increasing delay, i.e., then an $\incP$ algorithm can be constructed which uses polynomial space only.
Here, we show, on the downside, that for the general problem without restricting the team size an enumeration algorithm running in polynomial space cannot exist.

\section{Preliminaries}
Further, the underlying machine concept will be RAMs as we require data structures with logarithmic costs for standard operations. 
A detailed description of the RAM computation model may be found in \cite{L90}. 
The space occupied by a RAM is given by the total amount of used registers, provided that the content of each register is polynomially bounded in the size of the input.
Furthermore, we will follow the notation of Durand et~al.\ \cite{D16}, Creignou et~al.\ \cite{cmmjv16} and Schmidt \cite{S09}.
The complexity classes of interest are $\p$ and $\np$ (over the RAM model which is equivalent to the standard model over Turing machines in this setting). 

\paragraph*{Team-based Propositional Logic} \label{sec:team_based_propositional_logic}
Let $\mathcal{V}$ be a (countably infinite) set of variables. 
The class of all \emph{Poor Man's Propositional formulas} $\pl$ is derived via the grammar 
$$\varphi \defdefeq x \mid \neg x \mid 0 \mid 1 \mid \varphi \wedge \varphi,$$
where $x\in\mathcal{V}$. 
The set of all variables occurring in a propositional formula $\varphi$ is denoted by $\var(\varphi)$. 

Now we will specify the notion of teams and its interpretation on propositional formulas. 
An \emph{assignment} over $\mathcal{V}$ is a mapping $s\colon\mathcal{V} \rightarrow \{0,1\}.$	
We set  $\setassign{\mathcal{V}} \defeq \{s : s \text{ assignment over } \mathcal{V}\}.$
A \emph{team} $T$ over $\mathcal{V}$ is a subset $T\subseteq \setassign{\mathcal{V}}$. 
Consequently, the set of all teams over $\mathcal{V}$ is denoted by $\setteam{\mathcal{V}}$. 
If $X$ is a subset of $\mathcal{V}$, we set $T\restrict{X} \defeq \left\{s\restrict{X} : s\in T\right\},$ where $s\restrict{X}$ is the restriction of $s$ on $X$. 
If $T$ has cardinality $k\in\natural$, we say that $T$ is a \emph{$k$-Team}. 
If $\varphi$ is a formula, then a team (assignment) over $\var(\varphi)$ is called a \emph{team} (\emph{assignment}) \emph{for $\varphi$}.

A \emph{team-based prop\-o\-si\-tion\-al formula} $\varphi$ is constructed by the rule set of $\pl$ with the extension $\varphi \defdefeq \dep{P,Q} \iftoggle{classneg}{\mid \classneg \varphi},$
	where $P,Q$ are sets of arbitrary variables. 
We write $\dep{x_1,x_2,\dots,x_n}$ as a shorthand for $\dep{\{x_1,x_2,\dots,x_{n-1}\},\{x_n\}}$ and set $\pdl \defeq \pl(\dep{\cdot})$ for the formulas of \emph{Poor Man's Propositional Dependence Logic}.

\begin{definition}[Satisfaction]
	Let $\varphi$ be a team-based propositional formula and $T$ be a team for $\varphi$. 
	We define $T \models \varphi$ inductively by
	\begin{alignat*}{2}
		T& \models x &&\defequiv s(x)=1\quad\forall s\in T,\\
		T& \models \neg x &&\defequiv s(x)=0\quad\forall s\in T,  \\
		T& \models 1 &&\defequiv \text{true},\\
		T& \models 0 &&\defequiv T = \emptyset,  \\
		T& \models \varphi \wedge \psi &&\defequiv T \models \varphi \text{ and } T \models \psi,\\
\iflong		T& \models \varphi \classvee \psi &&\defequiv T \models \varphi \text{ or } T \models \psi, \\\fi
		T& \models \dep{P,Q} &&\defequiv \forall s,t\in T: s\restrict{P} = t\restrict{P} \Rightarrow s\restrict{Q} = t\restrict{Q}\iftoggle{classneg}{,}{.} 
	\end{alignat*}
	We say that \emph{$T$ satisfies $\varphi$} iff $T\models\varphi$ holds.
\end{definition}

Note that we have $T\models (x \wedge \neg x)$ iff $T=\emptyset$. This observation motivates the definition for $T\models 0$.
Observe that the evaluation in classical propositional logic occurs as the special case of evaluating singletons in team-based propositional logic.

\begin{definition}[Downward closure]
	\label{def:pl_downward_closure}
A team-based propositional formula $\varphi$ is called \emph{downward closed}, if for every team $T$ we have that 
		$T \models \varphi \Rightarrow \forall S\subseteq T: S \models \varphi.$
An operator $\circ$ of arity $k$ is called \emph{downward closed}, if $\circ(\varphi_1,\dots,\varphi_k)$ is downward closed for all downward closed formulas $\varphi_i$, $i=1,\dots,k$.
A class $\phi$ of team-based propositional formulas is called \emph{downward closed}, if all formulas in $\phi$ are downward closed.
\end{definition}

The following lemma then is straightforward to prove.

\begin{lemma}
	All atoms and operators in $\pdl$ are downward closed. In particular, $\pdl$ is downward closed.
\end{lemma}
\iflong
\begin{proof}
	It is easy to see that the atoms $x$, $\neg x$, $0$, $1$, $\dep{\cdot}$ are downward closed. Let $\varphi$, $\psi$ be two downward closed formulas. Let $T$ be a team with $T \models \varphi \wedge \psi$. Then we have $T\models \varphi$ and $T\models\psi$, so that $S\models\varphi$ and $S\models\psi$ for every subset $S\subseteq T$. It follows that $S\models\varphi\wedge\psi$. For this reason $\wedge$ is downward closed. The proof for $\classvee$ is similar.
\end{proof}
\fi

\paragraph*{Enumeration problems}
\label{sec:enumeration_problem}

Let $\Sigma$ be a finite alphabet and $(S,\ \le)$ a partially ordered set of possible solutions. 
An \emph{enumeration problem} is a triple $E = (Q,\ \sol,\ \le)$ such that
(i) $Q\subset \Sigma^*$ is a decidable language and
(ii) $\sol\colon Q \rightarrow \pow{S}$ is a computable function.
For an element $x\in Q$ we call $x$ an \emph{instance} and $\sol(x)$ its set of \emph{solutions}. 
If $\le$ is the trivial poset given by $x\le y \defequiv x = y,$ we omit it and write $E = (Q,\ \sol).$
Analogously, we write  $x<y$ for $x\leq y \text{ and } x\neq y.$

\begin{definition}[Enumeration algorithm]
	Let $E = (Q,\ \sol,\ \le)$ be an enumeration problem. A deterministic algorithm $\mathcal{A}$ is an \emph{enumeration algorithm} for $E$ if for every input $x\in Q$%
\iflong
	\begin{enumerate}
		\item $\mathcal{A}$ terminates,
		\item $\mathcal{A}$ outputs the set $\sol(x)$ without duplicates,
		\item for every $s,t\in \sol(x)$ with $s < t$ the solution $s$ is outputted before $t$.
	\end{enumerate}
\fi
\end{definition}


\begin{definition}[Delay]
	Let $\mathcal{A}$ be an enumeration algorithm for the enumeration problem $E = (Q,\ \sol,\ \le)$ and $x\in Q$. The \emph{$i$-th delay} of $\mathcal{A}$ is defined as the elapsed time between outputting the $i$-th and $(i+1)$-th solution of $\sol(x)$, where the $0$-th and $(|\sol(x)|+1)$-st delay are considered to happen  at the start and the end of the computation respectively. The $0$-th delay is called \emph{precomputation phase} and the $(|\sol(x)|+1)$-st delay is called \emph{postcomputation phase}.
\end{definition}

\begin{definition}
	Let $E = (Q,\ \sol,\ \le)$ be an enumeration problem and $\mathcal{A}$ be an enumeration algorithm for $E$. $\mathcal{A}$ is 
	\begin{enumerate}
		\item an \emph{$\incP$-algorithm} if there exists a polynomial $p$ such that the $i$-th delay on input $x\in Q$ is bounded by $p(|x| + i)$.
		\item a \emph{$\delayP$-algorithm} if there exists a polynomial $p$ such that all delays on input $x\in Q$ are bounded by $p(|x|)$.
		\item a \emph{$\delaySpaceP$-algorithm} if it is a $\delayP$-algorithm using polynomial amount of space with respect to the size of the input.
	\end{enumerate}
\end{definition}

For ease of notation, we define the classes $\delayP$ ($\incP,\delaySpaceP$) as the class of all enumeration problems admitting a $\delayP$- ($\incP$, $\delaySpaceP$)-algorithm.
%
%
%
%
%
%
%
Now we introduce the parametrised version of enumeration problems. The extensions are similar to those when extending $\p$ to $\fpt$. We follow Creignou et~al.\ \cite{cmmjv16}.

\begin{definition}[Parametrised enumeration problem]
	An enumeration problem $(Q, \sol, \le)$ together with a polynomial time computable \emph{parametrisation} $\kappa\colon\Sigma^* \rightarrow \natural$ is called a \emph{parametrised enumeration problem} $E = (Q,\ \kappa,\ \sol,\ \le)$. As before, if $\le$ is omitted, we assume $\le$ to be trivial.
\end{definition}


\begin{definition}
	Let $\mathcal{A}$ be an enumeration algorithm for a parametrised enumeration problem $E = (Q,\ \kappa,\ \sol,\ \le)$. If there exist a polynomial $p$ and a computable function $f\colon\natural\rightarrow\natural$ such that the $i$-th delay on input $x\in Q$ is bounded by $f(\kappa(x)) \cdot p(|x| + i)$, then $\mathcal{A}$ is an \emph{$\incFPT$-algorithm}. We call $\mathcal{A}$ a $\delayFPT$-algorithm if all delays on input $x\in Q$ are bounded by $f(\kappa(x)) \cdot p(|x|)$. The class $\incFPT$ contains all enumeration problems that admit an $\incFPT$-algorithm. The class $\delayFPT$ is defined analogously.
\end{definition}

	

\paragraph*{Group action}

The following section provides a compact introduction in group actions on sets. For a deeper introduction see, for instance, Rotman's textbook \cite{R95}.

\begin{definition}[Group action]\label{remark_group_actions}
	Let $G$ be a group with identity element $e$ and $X$ be a set. A \emph{group action of $G$ on $X$}, denoted by $G\acts X$, is a mapping $G\times X \rightarrow X$, $(g, x) \mapsto gx$, with 
	\begin{enumerate}[label={\upshape(\roman*)}]
		\item $ex = x\quad \forall x\in X$
		\item $(gh)x = g(hx)\quad \forall g,h\in G,\ x\in X.$
	\end{enumerate}
\end{definition}

Now observe the following. Let $G$ be a group and $X$ a set.
The mapping $(g,h)\mapsto gh$ for $g,h\in G$ defines a group action of $G$ on itself.
A group action $G\acts X$ induces a group action of $G$ on $\pow{X}$ by
$gS \defeq \{gs : s\in S\}$ for all $g\in G,\ S\subseteq X.$
Note that this group action preserves the cardinality of sets.

\begin{definition}[Orbit]
	Let $G\acts X$ be a group action and $x\in X$. Then the \emph{orbit of $x$} is given by
	$Gx \defeq \{gx : g\in G\}\subseteq X.$
\end{definition}

\begin{proposition}[{\cite{R95}}] 
	\label{prop:orbit_partition}
	Let $G\acts X$ be a group action and $x,y\in X$. Then either $Gx = Gy$ or $Gx \cap Gy = \emptyset$. Consequently the orbits of $G\acts X$ partition the set $X$.
\end{proposition}

\begin{definition}[Stabiliser]
	Let $G\acts X$ be a group action and $x\in X$. The \emph{stabilizer subgroup of $x$} is given by
	$G_x \defeq \{g\in G : gx=x \}$ and indeed is a subgroup of $G$.
\end{definition}

\begin{proposition}[Orbit-Stabiliser theorem, {\cite[Theorem 3.19]{R95}}]
	\label{prop:orbit_stabilizer}
	Let $G$ be a finite group acting on a set $X$. Let $x\in X$. Then the mapping $gG_x \mapsto gx$ is a bijection from $G/G_x$ to $Gx$. In particular, we have that $|Gx|\cdot |G_x| = |G|$.
\end{proposition}

\begin{proposition}[Cauchy-Frobenius lemma, {\cite[Theorem 3.22]{R95}}]
	\label{prop:cauchy_frobenius}
	Let $G$ be a finite group acting on a set $X$. Then the amount of orbits is given by
	$\frac{1}{|G|}\sum_{g\in G} |\{x\in X: gx=x\}|.$
\end{proposition}

\section{Results}
\label{chap:pl}

In this section, we investigate the complexity of enumerating all satisfying teams for various fragments of team-based propositional logic. 
After introducing the problem \enumteampl{} and its parametrised version \penumteampl{} we develop two enumeration algorithms for $\pdlnovee$, either guaranteeing polynomial delay or incremental delay in polynomial space. 

\begin{problem}
	\label{problem_enumteam}
	Let $\Phi$ be a class of team-based propositional formulas and $f\colon\natural \rightarrow \natural$ be a computable function. Then we define 
	$\enumteampl(\Phi,\ f) \defeq (\Phi,\ \sol)$
	where
	$$\sol(\varphi) \defeq \left\{\emptyset \neq T\in\setteam{\var(\varphi)} : T\models \varphi,\ |T| \leq f(|\varphi|)\right\}\quad \text{ for } \varphi\in\Phi.$$
\end{problem}
As we are interested in non-empty teams as solutions, we excluded the $\emptyset$ from the set of all solutions.
Nevertheless, formally by the \emph{empty team property}, it always holds that $\emptyset\models\varphi$.

\begin{problem}
Let $\Phi$ be a class of team-based propositional formulas and $f\colon\natural \rightarrow \natural$ a computable function. Then
	$\penumteampl(\phi) \defeq (\Phi\times\natural,\ \kappa,\ \sol)$
	where $\kappa((\varphi, k)) \defeq k$ and
	\begin{align*}
		\sol((\varphi, k)) \defeq \left\{\emptyset \neq T\in\setteam{\var(\varphi)} : T\models \varphi,\ |T| \leq k\right\} \text{ for }(\varphi,k)\in\Phi\times\natural.
	\end{align*} 
\end{problem}
We write \enumteampl($\Phi$) for \enumteampl($\Phi$, $n \mapsto 2^n$). Since $|T| \leq 2^{|\varphi|}$ holds for every team $T$ for $\varphi$, we effectively eliminate the cardinality constraint.
 As we shall see, the order in which the teams are outputted plays an important role in the following reasoning. There are two natural orders on teams to consider.

\begin{definition} [Order of cardinality]
	\label{defi:size_order}
	Let $R,S$ be two teams. Then we define a partial order on the set of all teams by
	$R \leq_\text{size} S \defequiv |R| < |S| \text{ or } R=S.$
\end{definition}

When a formula $\varphi$ is given, we assume to have a total order $\leq$ on $\setassign{\var(\varphi)}$ such that comparing two elements is possible in $\mathcal{O}(|\var(\varphi)|)$ and iterating over the set of all assignments is feasible with delay $\mathcal{O}(|\var(\varphi)|)$. When interpreting each assignment as a binary encoded integer, we obtain an appropriate order on $\setassign{\var(\varphi)}$ by translating the order on $\natural_0$. If necessary, one could demand that adjacent assignments differ in only one place by using the order induced by the Gray code. Now we are able to define the second order.

\begin{definition} [Lexicograph.\ order]
	\label{defi:lex_order}
	Let $R=\{r_1,\dots,r_n\},$ $S=\{s_1,\dots,s_m\}$ be two teams such that $r_1 < \cdots < r_n$ and $s_1 < \cdots < s_m$. Let $i$ be the maximum over all $j\in \natural_0$ such that $j\leq \min(n,m),\ r_\ell = s_\ell$ for all $\ell\in\{1,\dots ,j\}$. Then we define a partial order on $\setteam{\var(\varphi)}$ by
	$$R \leq_\text{lex} S \defequiv
	\begin{cases}
		n \leq m, & i = \min(n,m) \\
		r_{i+1} < s_{i+1}, & \text{else}.
	\end{cases}$$
\end{definition}

Observe that the lexicographical order is a total order that does not extend the order of cardinality. 
For example, we have $\{00, 01, 10\}<_\text{lex}\{00,10\}$ when assignments are ordered according to their integer representation.

\begin{problem} 
	Let $\Phi$ be a class of team-based propositional formulas and $f\colon\natural \rightarrow \natural$ be a computable function. We define
	$\enumteamplsize(\Phi,\ f) \defeq (\Phi,\ \sol,\ \leq_\text{size})$
	with $\sol$ as in Problem \ref{problem_enumteam}. \penumteamplsize{} is defined accordingly.
\end{problem}

\subsection{Enumeration in Poor Man's Propositional Dependence Logic}
\label{sec:enumerating_pdlnovee}

Now, we start with the task of enumerating satisfying teams for the fragment $\pdlnovee$, i.e., Poor Man's Propositional Logic.
The delay of the resulting algorithm is polynomial regarding the size of the input and the maximal size of an  outputted teams. 
As teams may grow exponentially large according to the input size, the delay will not be polynomial in the classical sense of \delayP{}. 
As a result, we proceed to \delayFPT{} and set the maximal cardinality of outputted teams as the parameter. 
Note that the drawback of having a polynomial delay in the output is minor. 
When following algorithms process the outputted teams, they have to input them first, requiring at least linear time in the output size.

In fact, we will see that we cannot obtain a \delayP-algorithm when the output is sorted by cardinality. 
This sorting, however, is an inherent characteristic of our algorithm as satisfying teams of cardinality $k$ are constructed by analysing those of cardinality $k-1$.

Before diving into details, we would like to introduce some notation used in this section.
	Let $\varphi\in\pdlnovee$ be fixed, $k\in\natural_0$, 
\begin{align*}
		n &\defeq |\var(\varphi)|, \\
		\setsatteams{k} &\defeq \left\{T\in\setteam{\var(\varphi)}: T\models\varphi, |T|=k\right\}, \\
		\setsatteamsz{k} &\defeq \left\{T\in\setsatteams{k}: (\forall x\in\var(\varphi): x\mapsto 0)\in T \right\}, \\
		\cardsatteams{k} &\defeq |\setsatteams{k}|, \\
		\cardsatteamsz{k} &\defeq |\setsatteamsz{k}|.
\end{align*}
	An assignment $s\in\setassign{\var(\varphi)}$ is depicted as a sequence of $0$ and $1$, precisely:
	$s = s(x_1)s(x_2)\dots s(x_n).$

\begin{example}
	For $\varphi \defeq \dep{x_1,x_2}$ we have: $n = 2$ and consequently 
	\begin{align*}
		\setsatteams{2} &= \{ \{00,10\}, \{00,11\}, \{01,10\}, \{01,11\} \},\\ 
		\setsatteamsz{2} &=  \{ \{00,10\}, \{00,11\} \},\\
		\setsatteams{3} &= \setsatteamsz{3} = \emptyset.
	\end{align*}
\end{example}


Note that formulas of the form $\varphi \equiv \left(\bigwedge_{x\in I} x\right) \wedge \left(\bigwedge_{x\in J} \neg x\right) \wedge \left( \bigwedge_{\ell\in L} \dep{P_\ell,Q_\ell}\right)$ can be simplified w.l.o.g.\ to 
\begin{align*}
\varphi' \defeq \bigwedge_{\ell\in L} \dep{P_\ell',Q_\ell'}\quad\text{with}\quad P_\ell' \defeq P_\ell\setminus(I\cup J),\ Q_\ell' \defeq Q_\ell\setminus(I\cup J). \tag{$\star$}	
\end{align*}
Then all satisfying teams for $\varphi$ can be recovered by extending those for $\varphi'$. For instance, the formula $$x_3 \wedge \dep{\{x_1\}, \{x_2,x_3\}} \wedge \dep{\{x_4\}, \{x_2, x_3\}}$$ may be reduced to $\dep{x_1, x_2} \wedge \dep{x_4, x_2}$. The team $\{00\!-\!0,00\!-\!1\}$ satisfies the latter formula (`$-$' indicates the missing $x_3$) and is extended to $\{0010,0011\}$ in order to satisfy the former one.\label{ex:robot_soccer_reduced}

\subsubsection{The group action of flipping bits}
By the semantics of $\dep{\cdot}$ we see that flipping the bit at a fixed position in all assignments of a team $T$ is an invariant for $T\models \dep{P, Q}$. For example, the teams $\{00,10\}$ and $\{00,11\}$ satisfy $\dep{x_1,x_2}$. The remaining $2$-teams satisfying the formula are given by $\{01,11\}$ and $\{01,10\}$. Note that these teams may be constructed from the previous ones by flipping the value of $x_2$. Accordingly, it would be enough to compute the satisfying teams $\{00,10\}$ and $\{00,11\}$, constructing the other $2$-teams by flipping bits. The concept of computing a minor set of satisfying $k$-Teams and constructing the remaining ones by flipping bits is the main concept of our algorithm for ensuring \fpt{}-delay.
	
	By identifying each assignment $s$ with the vector $(s(x_1),\dots,s(x_n))$ we obtain a bijection of sets $\ftwo^n \leftrightarrow \setassign{\var(\varphi)}.$
	We will switch between interpreting an element as an assignment or an $\ftwo$-vector as necessary, leading to expressions like $s+t$ for assignments $s$ and $t$. Those may seem confusing at first, but become obvious when interpreting $s$ and $t$ as vectors. Vice versa, we will consider $\ftwo$-vectors as assignments that may be contained in a team. When both notations are to be used, this is indicated by taking $s\in\ftwo^n \cong \setassign{\var(\varphi)}$ instead of simply writing $s\in\ftwo^n$ or $s\in\setassign{\var(\varphi)}$.

\begin{definition}[Group action of flipping bits]
	By the observation after Def.~\ref{remark_group_actions} the group action of ($\ftwo^n,+$) on itself induces a group action of $\ftwo^n$ on $\pow{\ftwo^n}$. On that account we obtain a group action $\ftwo^n\acts\setteam{\var(\varphi)}$, called \emph{group action of flipping bits}.
\end{definition}

Let $e_i$ be the $i$-th standard vector of $\ftwo^n$. 
Then the operation of $e_i$ on $\setteam{\var(\varphi)}$ corresponds to flipping the value for $x_i$ in each assignment of a team.

\begin{theorem}
	\label{theorem:orbit_invariance}
	Let $k\in\natural$. The restriction of $\ftwo^n\acts\pow{\ftwo^n}$ on $\setsatteams{k}$ yields a group action $\ftwo^n\acts\setsatteams{k}$.
\end{theorem}
\begin{proof}
	As the axioms of group actions still hold on a subset of $\pow{\ftwo^n}$, it remains to show that
	$zT \in \setsatteams{k}\quad \forall z\in\ftwo^n,\ T\in\setsatteams{k}.$
	Let $z\in\ftwo^n$ and $T\in\setsatteams{k}$. By the remark following Definition~\ref{remark_group_actions} we have $|zT| = k$. Let $P\subseteq\var(\varphi)$ and $s,t\in\setassign{\var(\varphi)}$. If $s', t'\in\setassign{\var(\varphi)}$ arise from $s,t$ by flipping the value for a variable $x_i$, then obviously
	$s\restrict{P} = t\restrict{P} \Leftrightarrow s'\restrict{P} = t'\restrict{P}.$
	It follows that
	$T \models \dep{P,Q} \Leftrightarrow zT \models \dep{P,Q}$ for all $P,Q\subseteq\var(\varphi).$
	When assuming that $\varphi$ has the form of ($\star$), it clearly holds that $zT \models \varphi$ because of $T\models \varphi$. This proves $zT \in \setsatteams{k}$.
\end{proof}

\begin{lemma}
	\label{lemma:representative_system_in_tzero}
	Let $T\in\setsatteams{k}$, $k\in\natural$. Then, we have that
	$\ftwo^n T \cap \setsatteamsz{k} \neq \emptyset.$
	For this reason $\setsatteamsz{k}$ contains a representative systems for the orbits of $\ftwo^n\acts\setsatteams{k}$.
\end{lemma}
\begin{proof}
	Take $s\in T\subseteq\setassign{\var(\varphi)}\cong \ftwo^n$. Then $sT \in \setsatteamsz{k}$ because of $z + z = \vec{0}$ for all $z\in\ftwo^n$.
\end{proof}

The previous lemma states that we can compute $\setsatteams{k}$ from $\setsatteamsz{k}$ by generating orbits. Next we want to present and analyse an algorithm for enumerating those orbits. The results are given in Theorem \ref{theorem:generating_orbits}.

\begin{definition}
	Let $\vec{0}\neq s = (s_1,\dots,s_n)\in\ftwo^n$ and $\mathcal{B}\subseteq\ftwo^n\setminus\{\vec{0}\}$. Then we define
\iflong
	\begin{align*}
	\last(s) &\defeq \max \left\{i\in\{1,\dots,n\} : s_i = 1\right\}, \\
	\last(\mathcal{B}) & \defeq \{\last(s) : s\in \mathcal{B}\}.
	\end{align*}
\fi
\end{definition}

\begin{definition}
	Let $\mathcal{B}$ be a subset of $\ftwo^n$. Then the subspace generated by $\mathcal{B}$ is defined by	
\iflong$$\text{span}(\mathcal{B}) \defeq \{b_1+\cdots + b_r : r\in\natural_0,\ b_i\in\mathcal{B}\ \forall i\in\{1,\dots,r\}\}.$$
\fi
\end{definition}

\begin{lemma}
	\label{lemma:distinct_last_base}
	Let $U$ be a subspace of the $\ftwo$-vector space $\ftwo^n$. Let $\mathcal{B}\subseteq U\setminus\{\vec{0}\}$ be a maximal subset with
	\begin{equation}
	\label{equa:distinct_last_base}
	b \neq b' \Rightarrow \last(b)\neq\last(b')\quad\forall b,b'\in \mathcal{B}.
	\end{equation}
	Then $\mathcal{B}$ is a basis for $U$.
\end{lemma}
\begin{proof}
	First we show that any set $A\subseteq U\setminus\{\vec{0}\}$ satisfying (\ref{equa:distinct_last_base}) is linearly independent. We conduct an induction over $|A|$. For $|A|=1$ the claim is obvious. Because of (\ref{equa:distinct_last_base}) there exists an element $a_0\in A$ with $\last(a_0) > \last(a)$ for all $a_0\neq a\in A$. When considering the $\last(a_0)$-th component, clearly the equation
	$$a_0 = \sum_{a_0\neq a\in A} \lambda_a a, \quad \lambda_a\in\ftwo$$
	has no solution. As $A\setminus\{a_0\}$ is linearly independent by induction hypothesis, it follows that A is linearly independent.
	
	Now assume that $\mathcal{B}$ does not generate $U$. We take an element $s\in U\setminus\text{span}(\mathcal{B})$ with minimal $\last(s)$. 
	As $\mathcal{B}$ is a maximal subset fulfilling (\ref{equa:distinct_last_base}), we have that $\last(b) = \last(s)$ for a suitable element $b\in \mathcal{B}$. 
	But then $s-b\in U\setminus\text{span}(\mathcal{B})$ with $\last(s-b) < \last(s)$ contradicts the minimality of $s$.
\end{proof}

\begin{theorem}
	\label{theorem:generating_orbits}
	Let $T\in\setsatteams{k}$, $k\in\natural$. Then $\ftwo^nT$ can be enumerated with delay $\bigo(k^3n)$.
\end{theorem}
\begin{proof}
	W.l.o.g. let $T\in\setsatteamsz{k}$. Otherwise, consider the team $zT$ with an arbitrary $z\in T$. Note that $T$ may have a nontrivial stabilizer subgroup so that duplicates occur when simply applying each $z\in\ftwo^n$ to $T$. However, Proposition \ref{prop:orbit_stabilizer} states that we can enumerate the orbit of $T$ without duplicates when applying a representative system for $\ftwo^n/(\ftwo^n)_T$.
	
	When taking $\ftwo^n$ as a vector space over $\ftwo$, the subspaces of $\ftwo^n$ correspond to the subgroups of $(\ftwo^n,+)$. In view of this any basis for a complement of the stabilizer subgroup $(\ftwo^n)_T$ of $T$ in $\ftwo^n$ generates a representative system for $\ftwo^n/(\ftwo^n)_T$.
	
	Take a basis $\mathcal{B}$ of $(\ftwo^n)_T$ as in Lemma \ref{lemma:distinct_last_base}. Set
	$\mathcal{C} \defeq \{e_i : i\in\{1,\dots,n\}\setminus\last(\mathcal{B})\},$
	where $e_i$ denotes the $i$-th standard vector of $\ftwo^n$. By construction of $\mathcal{C}$ we can arrange the elements of $\mathcal{B}\cup\mathcal{C}$ so that the matrix containing these elements as columns has triangular shape with $1$-entries on its diagonal. Consequently $\mathcal{B}\cup \mathcal{C}$ is a basis for $\ftwo^n$ and $\mathcal{C}$ is a basis for a complement of $(\ftwo^n)_T$.
	Now it remains to construct $\mathcal{B}$ as desired. For $s\in\setassign{\var(\varphi)} \cong \ftwo^n$ we have
$$s\in(\ftwo^n)_T \Rightarrow sT = T \Rightarrow s = s + \vec{0} \in T.$$
As a result, we can compute $(\ftwo^n)_T$ by checking $sT = T$ for $|T|=k$ elements in $\ftwo^n$. In fact it is enough to check $sT\subseteq T$ as we have $|sT|=|T|$. We obtain $\mathcal{B}$ by inserting each element of $(\ftwo^n)_T\setminus\{\vec{0}\}$ preserving (\ref{equa:distinct_last_base}) into $\mathcal{B}$. This shows that Algorithm \ref{algo:enumerating_orbits} outputs $\ftwo^n T$ without duplicates. The delay is dominated by the precomputation phase (lines \ref{algline:enumerating_orbits_precalculation_start} to \ref{algline:enumerating_orbits_precalculation_end}), which is $\bigo(k^3n)$. Note that we sort the $k$ assignments of each team in ascending order before returning it.
\end{proof}
	\begin{algorithm}[tb]
		\caption{Enumerating orbits}
		\label{algo:enumerating_orbits}
		
		\KwIn{A team $T$ with $\vec{0}\in T$}
		\KwOut{The orbit $\ftwo^n T$ of $T$ where each outputted team is sorted}
		\BlankLine
		$\mathcal{B}_\last \leftarrow \emptyset$\Comment*[r]{Assume that $\mathcal{B}_\last$ is sorted}
		\label{algline:enumerating_orbits_precalculation_start}
		\For(\Comment*[f]{$< k$ iterations}){$\vec{0}\neq s\in T$}{
			\lIf(\Comment*[f]{$\bigo( n)$}){$\last(s)\in\mathcal{B}_\last$}{\Continue}
			failed $\leftarrow$ false\;
			\For(\Comment*[f]{$\leq k$ iterations}){$t\in T$}{
				\lIf(\Comment*[f]{$\bigo(kn)$}){$s+t\notin T$}{failed $\leftarrow$ true}
			}
			\lIf(\Comment*[f]{$\bigo(n)$}){\Not failed}{$\mathcal{B}_\last \leftarrow \mathcal{B}_\last \cup \{\last(s)\}$}
		}
		$\mathcal{C}_\last \leftarrow \{1,\dots,n\}\setminus\mathcal{B}_\last$\Comment*[r]{$\bigo(n)$}
		\label{algline:enumerating_orbits_precalculation_end}
		\For{$s\in \text{span}(\{e_i : i\in\mathcal{C}_\last\})$}{
			Compute $sT$\Comment*[r]{$\bigo(kn)$}
			Sort $sT$\Comment*[r]{$\bigo(kn \log k)$}
			\Output $sT$\;
		}
	\end{algorithm}	
\iflong
\begin{example}
	Let $n=3$ and $T=\{000,100,010,110\}$. Note that $T$ satisfies the reduced formula from page~\pageref{ex:robot_soccer_reduced}. We compute the orbit $\ftwo^3 T$ of $T$ by algorithm \ref{algo:enumerating_orbits}. We check $sT=T$ for all nonzero assignments $s$ in $T$:
	\begin{align*}
		100 \cdot T &= \{100,000,110,010\} = \{000,100,010,110\} = T, \\
		010 \cdot T &= \{010,110,000,100\} = \{000,100,010,110\} = T, \\
		110 \cdot T &= \{110,010,100,000\} = \{000,100,010,110\} = T.
	\end{align*}
	On that account we obtain
	\begin{align*}
		\mathcal{B}_\last &= \{\last(100), \last(010), \last(110)\} = \{1,2\}, \\
		\mathcal{C}_\last &= \{3\}, \\
		\text{span}(\{e_i : i\in\mathcal{C}_\last\}) &= \{000, 001\}.
	\end{align*}
	Then the orbit of $T$ is given by
	$000 \cdot T = \{000,100,010,110\}$ and
	$001 \cdot T = \{001,101,011,111\}.$
\end{example}\fi

Finally we would like to relate $\cardsatteams{k}$ to $\cardsatteamsz{k}$. The larger the quotient $\cardsatteams{k} / \cardsatteamsz{k}$, the more computation costs are saved by generating orbits instead of computing $\setsatteams{k}$ immediately.

\begin{theorem}
	\label{theorem:teams_per_zero_team}
	Let $k\in\natural$ with $\cardsatteams{k}\neq 0$. Then, we have that
	$\cardsatteams{k}/\cardsatteamsz{k} = 2^n/k.$
\end{theorem}
\begin{proof}
	Because of $\cardsatteams{k}\neq 0$ and Lemma \ref{lemma:representative_system_in_tzero} it follows that $\cardsatteamsz{k}\neq 0$. For this reason we can choose $T\in\setsatteamsz{k}$. We claim
	\begin{equation}
	\label{equa:orbit_zero_intersection_cardinality}
	|\ftwo^n T \cap \setsatteamsz{k}| = \frac{k}{|(\ftwo^n)_T|}.
	\end{equation}
	For any $s\in\setassign{\var(\varphi)} \cong \ftwo^n$ we have that
	\begin{equation}
	\label{equa:orbit_zero_intersection}
	sT \in \setsatteamsz{k} \Leftrightarrow  \exists t\in T : s + t = \vec{0} \Leftrightarrow \exists t\in T : s = t \Leftrightarrow s\in T.
	\end{equation}
	Consequently we have 
	$\ftwo^n T \cap \setsatteamsz{k} = \left\{ sT : s\in T \right\} \eqdef TT.$
	Let $r,s\in T$. Both elements yield the same team $rT=sT$ iff $s \in r(\ftwo^n)_T$ so that for any fixed $r\in T$ we find exactly $|r(\ftwo^n)_T|=|(\ftwo^n)_T|$ ways of expressing $rT$ in the form of $sT$, where $s\in T$ by (\ref{equa:orbit_zero_intersection}). When iterating over the $k$ elements $sT$, $s\in T$, each team in $TT$ is counted $|(\ftwo^n)_T|$ times. It follows that
	$$|TT| = \frac{k}{|(\ftwo^n)_T|},$$
	proving (\ref{equa:orbit_zero_intersection_cardinality}).
	
	By Lemma \ref{lemma:representative_system_in_tzero} we find a representative system $R\subseteq\setsatteamsz{k}$ for the orbits of $\ftwo^n\acts \setsatteams{k}$. With Equation (\ref{equa:orbit_zero_intersection_cardinality}) and the Orbit-Stabilizer theorem (see Proposition \ref{prop:orbit_stabilizer}) we obtain
	\begin{alignat*}{2}
		\cardsatteams{k} 
		&= \sum_{T\in R} |\ftwo^n T| &&\qquad\text{(by Proposition \ref{prop:orbit_partition})} \\
		&= \sum_{T\in \setsatteamsz{k}} \frac{|\ftwo^n T|}{|\ftwo^n T \cap \setsatteamsz{k}|}\\
		&= \sum_{T\in \setsatteamsz{k}} \frac{|(\ftwo^n)_T|}{k} \cdot |\ftwo^n T| &&\qquad\text{(by (\ref{equa:orbit_zero_intersection_cardinality}))} \\
		&= \sum_{T\in \setsatteamsz{k}} \frac{|(\ftwo^n)_T|}{k} \cdot \frac{2^n}{|(\ftwo^n)_T|} &&\qquad\text{(by Proposition \ref{prop:orbit_stabilizer})} \\
		&= \frac{2^n}{k} \sum_{T\in \setsatteamsz{k}} 1\\
		&= \frac{2^n}{k} \cardsatteamsz{k}.
	\end{alignat*}
\end{proof}
\iflong
\begin{example}
	\label{ex:robot_soccer_orbits}
	Consider the reduced formula
	$\varphi \defeq \dep{x_1, x_3} \wedge \dep{x_2, x_3}$
	from page~\pageref{ex:robot_soccer_reduced}. Then the orbits of $\setsatteams{k}$, $k\in\natural$, and their corresponding stabiliser subgroups are given in Figure~\ref{fig:robot_soccer_orbits}. Teams located in $\setsatteamsz{k}$ are coloured red. Note that the amount of red teams in each orbit of $\setsatteams{k}$ matches $k$ divided by the cardinality of the stabiliser subgroup. Furthermore we have
	\begin{alignat*}{2}
		\frac{\cardsatteams{1}}{\cardsatteamsz{1}} = \frac{8}{1} = \frac{2^3}{1},\quad 
		\frac{\cardsatteams{2}}{\cardsatteamsz{2}} = \frac{16}{4} = \frac{2^3}{2},\quad 
		\frac{\cardsatteams{3}}{\cardsatteamsz{3}} = \frac{8}{3} = \frac{2^3}{3},\quad 
		\frac{\cardsatteams{4}}{\cardsatteamsz{4}} = \frac{2}{1} = \frac{2^3}{4}.
	\end{alignat*}
	
	\begin{figure}
\centering
\begin{tikzpicture}[   
  thick,
  team/.style={rounded corners, fill=blue!20, align=center, draw},
  teamz/.style={rounded corners, fill=red!40, align=center, draw},
  orbit/.style={very thick, draw},
  frame/.style={thick, draw},
  scale=0.93
]
	\coordinate (T1_1) at (-2.75,0);
	\coordinate (T2_1) at (-3,-4.75);
	\coordinate (T2_2) at (3,-4.75);
	\coordinate (T2_3) at (-3,-9);
	\coordinate (T2_4) at (3,-9);
	\coordinate (T3_1) at (0,-14.75);
	\coordinate (T4_1) at (3.5,0);
	
	\draw[frame] (-6.75,2.25) node[below right] {$\setsatteams{1}$} rectangle (1,-3);
	\draw[frame] (1,2.25) node[below right] {$\setsatteams{4}$} rectangle (6.75,-3);
	\draw[frame] (-6.75,-3) node[below right] {$\setsatteams{2}$} rectangle (6.75,-11.5);
	\draw[frame] (-6.75,-11.5) node[below right] {$\setsatteams{3}$} rectangle (6.75,-18.75);
	
	\newcommand*{\axisH}{2.75}
	\newcommand*{\axisV}{1.5}
	\newcommand*{\countElems}{8}
	
	\path[orbit] 
		(T1_1) ellipse({\axisH} and {\axisV}) 
		+(0,{-0.75-\axisV}) node {$(\ftwo^3)_T = \{\{000\}\}$};
	\path (T1_1)
		\foreach \t/\s [count=\i] in {
			$\{000\}$/teamz, 
			$\{100\}$/team, 
			$\{010\}$/team, 
			$\{110\}$/team, 
			$\{001\}$/team, 
			$\{101\}$/team, 
			$\{011\}$/team, 
			$\{111\}$/team} {
	      +({\axisH*cos(90-(\i-1)*360/\countElems)}, {\axisV*sin(90-(\i-1)*360/\countElems)}) node[\s] {\t}
		};
		
	\renewcommand*{\axisH}{1.75}
	\renewcommand*{\axisV}{1}
	\renewcommand*{\countElems}{4}
	
	\path[orbit] 
	(T2_1) ellipse({\axisH} and {\axisV}) 
	+(0,{-0.75-\axisV}) node {$(\ftwo^3)_T = \{\{000\},\{100\}\}$};
	\path (T2_1)
	\foreach \t/\s [count=\i] in {
		{$\{000,100\}$}/teamz, 
		{$\{010,110\}$}/team, 
		{$\{001,101\}$}/team, 
		{$\{011,111\}$}/team} {
		+({\axisH*cos(90-(\i-1)*360/\countElems)}, {\axisV*sin(90-(\i-1)*360/\countElems)}) node[\s] {\t}
	};
	
	\path[orbit] 
	(T2_2) ellipse({\axisH} and {\axisV})  
	+(0,{-0.75-\axisV}) node {$(\ftwo^3)_T = \{\{000\},\{010\}\}$};
	\path (T2_2)
	\foreach \t/\s [count=\i] in {
		{$\{000,010\}$}/teamz, 
		{$\{100,110\}$}/team, 
		{$\{001,011\}$}/team, 
		{$\{101,111\}$}/team} {
		+({\axisH*cos(90-(\i-1)*360/\countElems)}, {\axisV*sin(90-(\i-1)*360/\countElems)}) node[\s] {\t}
	};
	
	\path[orbit] 
	(T2_3) ellipse({\axisH} and {\axisV}) 
	+(0,{-0.75-\axisV}) node {$(\ftwo^3)_T = \{\{000\},\{110\}\}$};
	\path (T2_3)
	\foreach \t/\s [count=\i] in {
		{$\{000,110\}$}/teamz, 
		{$\{100,010\}$}/team, 
		{$\{001,111\}$}/team, 
		{$\{101,011\}$}/team} {
		+({\axisH*cos(90-(\i-1)*360/\countElems)}, {\axisV*sin(90-(\i-1)*360/\countElems)}) node[\s] {\t}
	};
	
	\path[orbit] 
	(T2_4) ellipse({\axisH} and {\axisV}) 
	+(0,{-0.75-\axisV}) node {$(\ftwo^3)_T = \{\{000\},\{111\}\}$};
	\path (T2_4)
	\foreach \t/\s [count=\i] in {
		{$\{000,111\}$}/teamz, 
		{$\{100,001\}$}/team, 
		{$\{010,101\}$}/team, 
		{$\{110,001\}$}/team} {
		+({\axisH*cos(90-(\i-1)*360/\countElems)}, {\axisV*sin(90-(\i-1)*360/\countElems)}) node[\s] {\t}
	};
	
	\renewcommand*{\axisH}{4}
	\renewcommand*{\axisV}{2.5}
	\renewcommand*{\countElems}{8}
	
	\path[orbit] 
	(T3_1) ellipse({\axisH} and {\axisV}) 
	+(0,{-0.75-\axisV}) node {$(\ftwo^3)_T = \{\{000\}\}$};
	\path (T3_1)
	\foreach \t/\s [count=\i] in {
		{$\{000,100,010\}$}/teamz, 
		{$\{100,000,110\}$}/teamz, 
		{$\{010,110,000\}$}/teamz, 
		{$\{110,010,100\}$}/team,
		{$\{001,101,011\}$}/team,
		{$\{101,001,111\}$}/team,
		{$\{011,111,001\}$}/team,
		{$\{111,011,101\}$}/team} {
		+({\axisH*cos(90-(\i-1)*360/\countElems)}, {\axisV*sin(90-(\i-1)*360/\countElems)}) node[\s] {\t}
	};
	
	\renewcommand*{\axisH}{1}
	\renewcommand*{\axisV}{0.75}
	\renewcommand*{\countElems}{2}
	
	\path[orbit] 
	(T4_1) ellipse({\axisH} and {\axisV}) 
	+(0,{-0.75-\axisV}) node {$(\ftwo^3)_T = \{\{000, 100, 010, 110\}\}$};
	\path (T4_1)
	\foreach \t/\s [count=\i] in {
		{$\{000,100,010,110\}$}/teamz, 
		{$\{001,101,011,111\}$}/team} {
		+({\axisH*cos(90-(\i-1)*360/\countElems)}, {\axisV*sin(90-(\i-1)*360/\countElems)}) node[\s] {\t}
	};
	
\end{tikzpicture}
\caption{Orbits of $\setsatteams{k}$ with $\varphi \defeq \dep{x_1, x_3} \wedge \dep{x_2, x_3}$.}
\label{fig:robot_soccer_orbits}
\end{figure}
\end{example}
\fi
\subsubsection{Constructing $\setsatteamsz{k}$}

Now that we are able to construct all satisfying $k$-teams from a representative system, the next step is the construction  of $\setsatteamsz{k}$. For this purpose the concept of coherence will prove useful. 

\begin{definition}[\phantom{}{\cite[Definition 3.1]{K13}}]
	Let $\phi$ be a team-based propositional formula. Then $\phi$ is \emph{$k$-coherent} iff for all teams $T$ we have that
	$$T\models \phi \ \Leftrightarrow\ R\models\phi\ \forall R\subseteq T \text{ with } |R|=k.$$
\end{definition}

\begin{proposition}[\phantom{}{\cite[Prop.\ 3.3]{K13}}]
	\label{prop:dep_coherent}
	The atom $\dep{\cdot}$ is $2$-coherent.
\end{proposition}

\begin{proposition}[\phantom{}{\cite[Prop.\ 3.4]{K13}}]
	\label{prop:conjunction_coherent}
	If $\phi$, $\psi$ are $k$-coherent then $\phi \wedge \psi$ is $k$-coherent.
\end{proposition}
Let $T=\{s_1,\dots,s_k\}$ be a team with $s_1<\dots<s_k$, $k\geq 2$. Then write 
\iflong
	\begin{align*}
	T_\red^1 \defeq \{s_1,\dots,s_{k-1}\},\ 
	T_\red^2 \defeq \{s_1,\dots,s_{k-2},s_k\},\  
	\max(T) \defeq s_k.
	\end{align*}
\fi
\noindent The following lemma provides a powerful tool for constructing the sets $\setsatteamsz{k}$.

\begin{lemma}
	\label{lemma:reduce_setsatteams}
	Let $T$ be as above and $k\defeq |T|\geq 3$. Then the following are equivalent:
	\begin{enumerate}[label={\upshape(\roman*)}]
		\item $T \in \setsatteamsz{k}$,
		\item $T_\red^1, T_\red^2 \in \setsatteamsz{k-1}$ and $\{\vec{0}, s_{k-1} + s_k\}\in \setsatteamsz{2}$.
	\end{enumerate}
\end{lemma}
\iflong
\begin{proof}
	After simplifying $\varphi$ we may assume that $\varphi$ is a conjunction of dependence atoms. In particular $\varphi$ is $2$-coherent by Proposition \ref{prop:dep_coherent} and \ref{prop:conjunction_coherent}.
	
	(i) $\Rightarrow$ (ii): Let $T\in \setsatteams{k}$. Any subset of cardinality 2 contained in $T_\red^1$ or $T_\red^2$ is a subset of $T$. The $2$-coherence of $\varphi$ yields $T_\red^i\models \varphi$ for $i\in\{1,2\}$. Furthermore $\vec{0}=s_1\in T_\red^i$ and $|T_\red^i| = k-1$ holds. This gives us $T_\red^1, T_\red^2\in \setsatteamsz{k-1}$. Again by the $2$-coherence  of $\varphi$ we obtain that $\{s_{k-1}, s_{k}\} \in \setsatteams{2}$. Applying the group action $\ftwo^n \acts \setsatteams{2}$ shows that $\{\vec{0}, s_{k-1} + s_k\}\in\setsatteamsz{2}$.
	
	(ii) $\Rightarrow$ (i): First note that $\vec{0}\in T_\red^1\subset T$ and $|T| = |T_\red^1| + 1 = k$. Assume $T\not\models \varphi$. Then by $2$-coherence there exists a subset $R\subseteq T$ with $|R|=2$ and $R\not\models\varphi$. In particular, we have $R\not\subseteq T_\red^1, T_\red^2$, implying $R=\{s_{k-1}, s_k\}$. This contradicts $s_{k-1}R = \{\vec{0}, s_{k-1} + s_k\}\in\setsatteamsz{2}$.
\end{proof}
\fi
\begin{algorithm}
	\caption{Constructing $\setsatteamsz{k}$}
	\label{algo:constructing_setsatteamsz}
	
	\KwIn{$k\in\natural$, $k\geq 2$}
	\KwDependencies{If $k>2$: $\mathcal{D}_2[\{\vec{0}\}]$, $\mathcal{D}_{k-1}$ of the previous iteration}
	\KwResult{$\setsatteamsz{k}$}
	\BlankLine
	
	$\setsatteamsz{k} \leftarrow \emptyset$\;
	\iflong $\mathcal{D}_k \leftarrow \New\ \Map(\text{Team}, \List(\text{Assignment}))$\;\fi
	\eIf{$k=2$}{
		$\mathcal{D}_2[\{\vec{0}\}] \leftarrow \emptyset$\;
		\For(\Comment*[f]{$\leq 2^n$ iterations}){$\vec{0}\neq s\in \setassign{\var(\varphi)}$}{
			\If(\Comment*[f]{$\bigo(|\varphi|)$}){$\{\vec{0}, s\}\models\varphi$}{
				$\mathcal{D}_2[\{\vec{0}\}] \leftarrow \mathcal{D}_2[\{\vec{0}\}] \cup \{s\}$\Comment*[r]{$\bigo(n)$}
				\label{algline:constructing_setsatteamsz_d2_insertion}
				$\setsatteamsz{2} \leftarrow \setsatteamsz{2} \cup \{\vec{0}, s\}$\Comment*[r]{$\bigo(n)$}
			}
		}
	}{
	\For{$(T,L)\in \mathcal{D}_{k-1}$}{
		\For(\Comment*[f]{$\cardsatteamsz{k-1}$ iterations}){$r\in L$}{ \label{algline:constructing_setsatteamsz_loop2}
			$T' \leftarrow T \cup \{r\}$, $\mathcal{D}_k[T'] \leftarrow \emptyset$\;
			\For(\Comment*[f]{$\leq 2^n$ iterations}){$s\in L$ with $s>r$}{ \label{algline:constructing_setsatteamsz_loop3}
				\If(\Comment*[f]{$\bigo(n)$}){$r+s\in\mathcal{D}_2[\{\vec{0}\}]$}{ \label{algline:constructing_setsatteamsz_if}
					$\mathcal{D}_k[T'] \leftarrow \mathcal{D}_k[T'] \cup \{s\}$\Comment*[r]{$\bigo(kn)$}
					$\setsatteamsz{k} \leftarrow \setsatteamsz{k} \cup \{T' \cup \{s\}\}$\Comment*[r]{$\bigo(kn)$}
				}
			}
		}
	}
}
\end{algorithm}

Algorithm \ref{algo:constructing_setsatteamsz} computes the sets $\setsatteamsz{k}$ by exploiting the previous lemma. In order to ensure fast list operations, we manage teams in tries \cite[chapter 6.3]{K98}. Since any team of cardinality $k$ may be described by $kn$ bits, the standard list operations as searching, insertion and deletion are realised in $\bigo(kn)$. We organise satisfying teams such that all teams of cardinality $k$ which only differ in their maximal assignment are described by a list $\mathcal{D}_k[T']$, where $T'$ is the team containing the common $k-1$ smaller assignments. It suffices to store the maximal assignment of each team $T$ described in $\mathcal{D}_k[T']$ since $T$ may be recovered by $T'$ and $\max(T)$. Hence $\mathcal{D}_k$ becomes a collection of lists indexed by teams of cardinality $k-1$. The following lemma states the correct construction of $\mathcal{D}_k$ in Algorithm \ref{algo:constructing_setsatteamsz}.

\begin{lemma}
	\label{lemma:algorithm_setsatteamsz}
	Let $k\geq 2$. For $T\in\setsatteamsz{k}$ we have that $\max(T)\in\mathcal{D}_k[T_\red^1]$. Vice versa, if $s\in\mathcal{D}_k[T]$, then it follows that $T\cup \{s\}\in\setsatteamsz{k}$ and $s>\max(T)$.
\end{lemma}\iflong
\begin{proof} We conduct an induction over $k$.
	
	Induction basis ($k=2$): Let $T\in\setsatteamsz{2}$. It follows that $T_\red^1 = \{\vec{0}\}$. As we have $\{\vec{0}, \max(T)\}\models\varphi$, in line \ref{algline:constructing_setsatteamsz_d2_insertion} $\max(T)$ is inserted into $\mathcal{D}_2[T_\red^1]$. Now let $s\in\setassign{\var(\varphi)}$ and $T\in\setteam{\var(\varphi)}$ such that $s\in\mathcal{D}_2[T]$. The only team occurring in $\mathcal{D}_2$ is $T=\{\vec{0}\}$. We have $s\in\mathcal{D}_2[T]$ iff $T \cup \{s\} = \{\vec{0}, s\}\models\varphi$ and $s\neq \vec{0}$. The claim follows.
	
	Induction step ($k-1\rightarrow k$): Let $T=\{s_1,\dots,s_k\}\in\setsatteamsz{k}$, $s_1<\dots<s_k$. By induction hypothesis and Lemma~\ref{lemma:reduce_setsatteams} we have that $s_{k-1}\in\mathcal{D}_{k-1}[T_\red^1 \setminus \{s_{k-1}\}]$. Accordingly, the loop body of line \ref{algline:constructing_setsatteamsz_loop2} is invoked with $T \corresponds T_\red^1 \setminus \{s_{k-1}\}$, $r \corresponds s_{k-1}$, $T' \corresponds T_\red^1$. Furthermore by Lemma \ref{lemma:reduce_setsatteams} the loop body of line \ref{algline:constructing_setsatteamsz_loop3} is invoked with $s\corresponds s_k$, passing the check in line \ref{algline:constructing_setsatteamsz_if}. As a result, $s_k = \max(T)$ is inserted into $\mathcal{D}_k[T_\red^1]$. 
	
	Now let $s\in\setassign{\var(\varphi)}$ and $T\in\setteam{\var(\varphi)}$ such that $s\in\mathcal{D}_k[T]$. Then by the construction of $\mathcal{D}_k$ there exist a team $T'$ and $r\in\mathcal{D}_{k-1}[T']$ with $r<s$, $s\in\mathcal{D}_{k-1}[T']$, $T'\cup\{r\} = T$ and $\{\vec{0}, r+s\}\in\setsatteamsz{2}$. The induction hypothesis yields $T' \cup \{r\}, T'\cup \{s\}\in\setsatteamsz{k-1}$ and $r > \max(T')$, implying $s > r = \max(T)$. As $s$ and $r$ are the largest elements of $T\cup\{s\}$, it follows that $(T\cup\{s\})_\red^1 = T' \cup \{r\}$ and $(T\cup\{s\})_\red^2 = T' \cup \{s\}$. By Lemma \ref{lemma:reduce_setsatteams} we obtain $T\cup\{s\}\in\setsatteamsz{k}$.
\end{proof}\fi

\iflong
\begin{corollary}
	Algorithm \ref{algo:constructing_setsatteamsz} constructs the sets $\setsatteamsz{k}$ correctly.
\end{corollary}
\begin{proof}
	Every team $T$ inserted into $\setsatteamsz{k}$ by Algorithm \ref{algo:constructing_setsatteamsz} has the form $T'\cup\{s\}$ with $s\in\mathcal{D}_k[T']$. Then by Lemma \ref{lemma:algorithm_setsatteamsz} it follows that $s>\max(T)$ and $T\in\setsatteamsz{k}$. Note that the decomposition of $T$ is unique because of $s>\max(T)$. On that account $T$ is inserted only once.
	
	Now let $T\in\setsatteamsz{k}$. Lemma \ref{lemma:algorithm_setsatteamsz} states that $\max(T) \in\mathcal{D}_k[T_\red^1]$. After inserting $\max(T)$ into $\mathcal{D}_k[T_\red^1]$, Algorithm \ref{algo:constructing_setsatteamsz} inserts $T_\red^1 \cup \max(T) = T$ into $\setsatteamsz{k}$.
\end{proof}

\begin{corollary}
	\label{coro:algorithm_setsatteamsz_complexity}
	Algorithm \ref{algo:constructing_setsatteamsz} requires $\cardsatteamsz{k-1}\cdot 2^n \cdot \bigo(k|\varphi|)$ time on input $k\in\natural$.
\end{corollary}
\begin{proof}
	Note
	$$\{\vec{0}, s\} \models \dep{P, Q}\ \Leftrightarrow\ s\models \left(\bigvee_{x\in P}x\right) \vee \left(\bigwedge_{y\in Q} \neg y\right)\quad\forall \vec{0}\neq s \in\setassign{\var(\varphi)}.$$
	As a consequence checking $\{\vec{0}, s\}\models\varphi$ can be accomplished in linear time by evaluating a $\mathcal{P\!L}$-formula of length $\bigo(|\varphi|)$ (where $\lor$ has the classical propositional disjunction semantics). Accessing the list $\mathcal{D}_k[T]$ for a team $T$ is in $\bigo(kn)$ if $\mathcal{D}_k$ is implemented as a trie. If that also is the case for the inner list $\mathcal{D}_k[T]$, its operations are realised in $\bigo(\log 2^n) = \bigo(n)$, which is contained in $\bigo(kn)$.
	
	As the decomposition for $T\in\setsatteamsz{k-1}$ into $T'$ and $s$ with $s>\max(T')$ is unique, applying Lemma \ref{lemma:algorithm_setsatteamsz} yields that the loop body of line \ref{algline:constructing_setsatteamsz_loop2} is invoked $\cardsatteamsz{k-1}$ times.
	
	Taking into account that $n\leq |\varphi|$, we obtain the claim by adding up all costs.
\end{proof}
\fi
\iflong
\begin{example}
	\label{ex:robot_soccer_setsatteamsz}
	We construct the sets $\setsatteamsz{k}$ for the reduced formula
	$\varphi \defeq \dep{x_1, x_3} \wedge \dep{x_2, x_3}$
	from page~\pageref{ex:robot_soccer_reduced}. Trivially, we have that $\setsatteamsz{1} = \{\{000\}\}$. When computing $\setsatteamsz{2}$, we have to identify all nonzero assignments that satisfy
	$(x_1 \vee \neg x_3) \wedge (x_2 \vee \neg x_3).$
	Obviously, the satisfying assignments are $100$, $010$, $110$ and $111$. We obtain
	$\mathcal{D}_2[\{000\}] = \{100,010,110,111\}$.
	Figure \ref{fig:robot_soccer_setsatteamsz} illustrates the construction of the remaining lists and the resulting sets $\setsatteamsz{k}$. We are able to verify that the orbits presented in Example~\ref{ex:robot_soccer_orbits} are exactly those of $\setsatteams{k}$, $k\in\natural$. Each orbit contains at least one element of $\setsatteamsz{k}$ and every team in $\setsatteamsz{k}$ can be recovered in one orbit of Figure \ref{fig:robot_soccer_orbits}.

\end{example}

\begin{figure}
\centering
\begin{tikzpicture}[   
  thick,
  assignment/.style={rectangle, fill=blue!20, align=center, draw},
  teamz/.style={rounded corners, fill=red!40, align=center, draw},
  orbit/.style={very thick, draw},
  frame/.style={rounded corners, fill=blue!10, thick, draw}
]

	\newcommand*{\Width}{9}
	
	\draw[frame] 	(-0.75,0.5) 
					node[right] at (-3.75,0) {$\mathcal{D} \defeq \mathcal{D}_2[\{000\}]$} 
					rectangle 
					({\Width+0.75},-0.5);
					
	\draw[frame] 	(-0.75,-2.5) 
					node[right] at (-3.75,-3) {$\mathcal{D}_3[\{000,100\}]$}
					rectangle 
					({\Width*2/5+0.75},-3.5);
					
	\draw[frame] 	({\Width*3/5-0.75},-2.5) 
					rectangle 
					({\Width*4/5+0.75},-3.5) 
					node[below] at +(-\Width/10-0.75,0) {$\mathcal{D}_3[\{000,010\}]$};
					
	\draw[frame] 	({\Width-0.75},-2.5) 
					rectangle 
					({\Width+0.75},-3.5) 
					node[below, align=right] at +(-0.75,0) {$\mathcal{D}_3[\{000,$\\$110\}]$};
					
	\draw[frame] 	({\Width/11-0.75},-5.5) 
					node[right] at (-3.75,-6) {$\mathcal{D}_4[\{000,100,010\}]$}
					rectangle 
					({\Width/11+0.75},-6.5);

	\foreach \t [count=\i] in {
		{$100$},
		{$010$},
		{$110$},
		{$111$}} {
		\node[assignment] (T2-\i) at ({\Width/3*(\i-1)},0) {\t};
	}
	
	\foreach \t/\p in {
		$010$/1,
		$110$/2,
		$110$/4} {
		\node[assignment] (T3-\p) at ({\Width/5*(\p-1)},-3) {\t};
	}
	
	\foreach \p in {
		3,
		5,
		6} {
		\coordinate (T3-\p) at ({\Width/5*(\p-1)},-3);
	}
	
	\node[assignment] (T4-1) at ({\Width/11},-6) {$110$};
	
	\foreach \r/\s/\e/\h/\j/\t in {
		1/2/1/2/3/$110\in\mathcal{D}$,
		1/3/2/2/3/$010\in\mathcal{D}$,
		2/3/4/2/3/$100\in\mathcal{D}$,
		1/2/1/3/4/$100\in\mathcal{D}$} {
		\draw[fill=black] (T\j-\e) + (0,1.5) circle(0.1) -- (T\h-\r);
		\draw (T\j-\e) -- +(0,1.5) node[below] {\t} -- (T\h-\s);
	}
	
	\foreach \r/\s/\e/\h/\j/\t in {
		1/4/3/2/3/$011\notin\mathcal{D}$,
		2/4/5/2/3/$101\notin\mathcal{D}$,
		3/4/6/2/3/$001\notin\mathcal{D}$} {
		\draw[fill=red,color=red] (T\j-\e) + (0,1.5) circle(0.1) -- (T\h-\r);
		\draw[color=red] (T\j-\e) -- +(0,1.5) node[below] {\t} -- (T\h-\s);
		\draw[ultra thick, color=red] (T\j-\e) +(-0.25,-0.25) -- +(0.25,0.25);
		\draw[ultra thick, color=red] (T\j-\e) +(0.25,-0.25) -- +(-0.25,0.25);
	}
	
	\node[right] (D1) 	at (-3.75,-7.5) {$\mathcal{D}_2[\{000\}]$};
	\node[right] (D2) at (-3.75,-8.5) {$\mathcal{D}_3[\{000,100\}]$};
	\node[right] (D3) at (-3.75,-9.5) {$\mathcal{D}_3[\{000,010\}]$};
	\node[right] (D4)	at (-3.75,-10.5) {$\mathcal{D}_4[\{000,100,010\}]$};
	
	\foreach \i in {1,2,3,4} {
		\draw[->] (D\i) -- (0.3,{-6.5-\i});
	}
	
	\foreach \t [count=\i] in {
	{$\{000,100\}$},
	{$\{000,010\}$},
	{$\{000,110\}$},
	{$\{000,111\}$}} {
		\node[teamz, right] at ({0.5 + 2*(\i-1)}, -7.5) {\t};
	}
	
	\foreach \t [count=\i] in {
		{$\{000,100,010\}$},
		{$\{000,100,110\}$}} {
		\node[teamz, right] at ({0.5 + 2.7*(\i-1)}, -8.5) {\t};
	}

	\node[teamz, right] at (0.5, -9.5) {$\{000,010,110\}$};
	\node[teamz, right] at (0.5, -10.5) {$\{000,100,010,110\}$};
	
	\draw[thick,decorate,decoration={brace}] (9,-7.2) -- (9,-7.8) node[midway, right] {$\setsatteamsz{2}$};
	\draw[thick,decorate,decoration={brace}] (9,-8.2) -- (9,-9.8) node[midway, right] {$\setsatteamsz{3}$};
	\draw[thick,decorate,decoration={brace}] (9,-10.2) -- (9,-10.8) node[midway, right] {$\setsatteamsz{4}$};
	
\end{tikzpicture}
\begin{tikzpicture}
	\node at (0,0) {$\setsatteamsz{k}=\emptyset\quad \forall k>4$};
\end{tikzpicture}
\caption{Construction of $\setsatteamsz{k}$ with $\varphi \defeq \dep{x_1, x_3} \wedge \dep{x_2, x_3}$.}
\label{fig:robot_soccer_setsatteamsz}
\end{figure}
\fi
\iflong
\subsubsection{The algorithm}
\label{sec:pl_pdlnovee_algorithm}
\fi

Although by Corollary \ref{coro:algorithm_setsatteamsz_complexity} Algorithm \ref{algo:constructing_setsatteamsz} does not perform in polynomial time on input $k\in\natural$, we can ensure polynomial delay when distributing its execution over the process of outputting all satisfying teams of cardinality $k-1$. For this reason we investigate the costs of computing $\setsatteamsz{k}$ divided by $\cardsatteams{k-1}$. With Corollary \ref{coro:algorithm_setsatteamsz_complexity} and 
\iflong$$k-1 = \frac{\cardsatteamsz{k-1} \cdot 2^n}{\cardsatteams{k-1}},$$
\fi
which is a transformation of the equation in Theorem \ref{theorem:teams_per_zero_team}, we obtain
\begin{align*}
	\frac{\text{computationCosts}(\setsatteamsz{k})}{\cardsatteams{k-1}}
	= \frac{\cardsatteamsz{k-1}\cdot 2^n \cdot \bigo(k|\varphi|)}{\cardsatteams{k-1}} 
	= (k-1)\cdot\bigo(k|\varphi|) 
	= \bigo(k^2|\varphi|).
\end{align*}

Since the delay of generating the orbits $\ftwo^n T$ is $\bigo(k^3n)$ by Theorem \ref{theorem:generating_orbits}, the overall delay of Algorithm \ref{algo:enumerating_satteams} is bounded by $\bigo(k^3|\varphi|)$. Note that the cost of removing elements in $\setsatteamsz{k}$, which is $\bigo(kn)$, is contained in $\bigo(k^3|\varphi|)$. 
In practise, we interleave both computation strands by executing $k$ iterations of the loop at line~\ref{algline:constructing_setsatteamsz_loop3} in Algorithm~\ref{algo:constructing_setsatteamsz} whenever a team is outputted.
\begin{algorithm}[tb]
	\caption{Enumerating satisfying teams in $\pdlnovee$, ordered by cardinality}
	\label{algo:enumerating_satteams}
	
	\KwIn{A team-based propositional formula $\varphi$ as in Equation ($\star$), $k\in\natural$}
	\KwOut{All teams $T$ for $\varphi$ with $T\models\varphi$, $1\leq|T|\leq k$}
	\BlankLine
	
	$\setsatteamsz{1} \leftarrow \{\{\vec{0}\}\}$\;
	\For{$\ell=2,\dots,k+1$}{
		\Parallel{
			\While{$\setsatteamsz{\ell-1}\neq \emptyset$}{
				Choose $T\in\setsatteamsz{\ell-1}$\;
\iflong
				\For{$T' \in \ftwo^n T$ (Algorithm \ref{algo:enumerating_orbits})}{
					\Output $T'$\;
					$\setsatteamsz{\ell-1} \leftarrow \setsatteamsz{\ell-1} \setminus \{T'\}$\;%
					\label{algline:enumerating_satteams_removal}%
				}
\fi				
			}
		}
			\textbf{simultaneously} Compute $\setsatteamsz{\ell}$ by Algorithm \ref{algo:constructing_setsatteamsz}\;
		\lIf{$\setsatteamsz{\ell} = \emptyset$}{\Break}
	}
\end{algorithm}
\iflong
\begin{theorem}
	Algorithm $\ref{algo:enumerating_satteams}$ enumerates all satisfying teams $T$ for $\varphi$ with $1\leq|T|\leq k$ without duplicates.
\end{theorem}
\begin{proof}
	 It is easy to see that all dependencies in Algorithm \ref{algo:constructing_setsatteamsz} and \ref{algo:enumerating_satteams} are resolved in time. By Proposition \ref{prop:orbit_partition} and Lemma \ref{lemma:representative_system_in_tzero} every satisfying team is outputted at least once. By removing every outputted element in line \ref{algline:enumerating_satteams_removal} no orbit is outputted twice, preventing duplicates.
\end{proof}\fi
Finally, we conclude.

\begin{theorem}
	\label{theorem:enumteamsize}
	\begin{enumerate}[label={\upshape(\roman*)}]
		\item $\penumteamplsize(\pdlnovee)\in\delayFPT$,
		\item $\enumteamplsize(\pdlnovee, f)\in\delayP$ for any poly.\ time computable function $f\in n^{\bigo(1)}$.
	\end{enumerate}
\end{theorem}

\subsubsection{Consequences of sorting by cardinality}

In the previous section we have seen that the restriction on polynomial teams is sufficient to obtain a $\delayP$-algorithm for $\pdlnovee$. As we will see in this section, the restriction is not only sufficient, but also necessary when the output is sorted by its cardinality. Consequently, the algorithm presented above is optimal regarding output size.


\begin{lemma}
	\label{lemma:enumteamsize_counterexample}
	Let $k\geq 2$ and 
	$\varphi(x_1,\dots,x_k) \defeq \bigwedge_{i=1}^{k-1} \dep{x_i, x_k} \in \pdlnovee.$
	Then for any team $T\neq \emptyset$ with $T\models\varphi$ and $|T|\geq 3$ we have that $\left|T\restrict{\{x_k\}}\right| = 1$.
\end{lemma}
\iflong
\begin{proof}
	Let $T$ be a team with $T\models\varphi$ and $|T|\geq 3$. Set
	$T\restrict{x_k=i} \defeq \{s\in T : s(x_k) = i\}$ for $i\in\{0,1\}$.
	Assume that $|T\restrict{\{x_k\}}| = 1$ does not hold. Then w.l.o.g. $T\restrict{x_k=0} \neq \emptyset$ and $|T\restrict{x_k=1}| > 1$.
	Take $r\in T\restrict{x_k=0}$. For any $s\in T\restrict{x_k=1}$ we have that $\{r, s\}\models\varphi$ because $\varphi$ is downward closed. In particular, $\{r,s\}\models \dep{x_i, x_k}$ for all $i\in\{1,\dots,k-1\}$, yielding $r(x_i) \neq s(x_i)$ because of $r(x_k) \neq s(x_k)$. On that account $s$ is uniquely determined by $r$, contradicting $\left|T\restrict{x_k=1}\right| > 1$.
\end{proof}
\fi
\begin{theorem}
	\label{theorem:enumteamsize_restriction}
	Let $f$ be a polynomial time computable function. Then we have that
\iflong$$\enumteamplsize(\pdlnovee, f)\in\delayP\text{ if and only if }f\in n^{\bigo(1)}.$$
\fi
\end{theorem}
\begin{proof}
	``$\Leftarrow$'': immediately follows from Theorem \ref{theorem:enumteamsize}.
	
	``$\Rightarrow$'': Let $f\notin n^{\bigo(1)}$. Assume that $\enumteamplsize(\pdlnovee, f)\in\delayP$ holds via an algorithm with a delay bounded by $n^c$, $c\in\natural$. Then there exists $k\in\natural$ such that 
\iflong$$z \defeq \min \{f(k), 2^{k-1}\} > 4^c\cdot k^c \geq k\geq 3.$$
\fi
	Let $\varphi$ be as in Lemma \ref{lemma:enumteamsize_counterexample}. Obviously, there exist teams $T_0, T_1\in\setsatteams{z}$ with $s(x_k) = i$ for all $s\in T_i$, $i\in\{0,1\}$. Since the elements in $\setsatteams{z}$ have to be outputted in succession and $\left|T\restrict{\{x_k\}}\right| = 1$ for any $T\in\setsatteams{z}$, we can choose $T_0$ and $T_1$ such that both teams are outputted in consecutive order. However, both teams differ in at least $z$ bits describing the evaluation at $x_k$. For this reason the delay is at least
	$z > (4k)^c \geq (|\varphi|)^c,$
	contradicting that the delay is bounded by $n^c$.
\end{proof}

\begin{corollary}
	$\enumteamplsize(\pdlnovee)\notin \delayP$.
\end{corollary}\iflong
\begin{proof}
	Since $(n\mapsto 2^n)\notin n^{\bigo(1)}$ the claim follows immediately from Theorem \ref{theorem:enumteamsize_restriction}.
\end{proof}\fi

The trick of examining the symmetric difference of consecutive teams gives rise to the previous theorem. Unfortunately this trick cannot be applied to arbitrary orders and certainly fails for the lexicographical order. In order to prove this claim, consider Theorem \ref{theorem:symmetric_difference_lex} with $S=\setassign{\var(\varphi)}$, $X=\left\{T\in\setteam{\var(\varphi)} : T\models \varphi\right\}$.

\begin{theorem}
	\label{theorem:symmetric_difference_lex}
	Let $S=\{s_1,\dots,s_n\}$ be a finite totally ordered set and $X\subseteq\mathcal{P}(S)$ be a downward closed set, meaning
	$T\in X \Rightarrow R\in X\ \forall R\subseteq T.$
	When $X$ is ordered lexicographically in respect with the order on $S$, the symmetric difference $\triangle$ between two consecutive elements in $X$ is at most $3$.
\end{theorem}\iflong
\begin{proof}
	W.l.o.g, we assume that $\{s_i\}\in X$ for all $i\in\{1,\dots,n\}$. We prove the claim by induction over $n$. The induction basis for $n=1$ is obvious. Consider the induction step from $n-1$ to $n$.
	
	Set $X_i \defeq \{T\in X : s_1,\dots,s_{i-1}\notin T,\ s_i\in T\},$	$\overline{X_i} \defeq \{T\setminus \{s_i\} : T\in X_i\}$ for $i\in \{1,\dots,n\}$. 
	The subsets $X_i$ together with $\{\emptyset\}$ form a partition of $X$. Furthermore it is easy to see that the sets $\overline{X_i}$ are downward closed and that $T < T'$ for all $T\in X_i$, $T'\in X_j$ with $i < j$.
	
	By induction hypothesis the symmetric difference for consecutive elements in $\overline{X_i}$ is at most $3$. As the order on $X_i$ corresponds to the order on $\overline{X_i}$, the same applies to $X_i$. It remains to check the consecutive elements located in different $X_i$'s. For any element $T = \{s_{i_1},\dots,s_{i_k}\} \in X_i$ with $s_{i_1} < \dots < s_{i_k}$ we have $T \geq \{s_{i_1}\} = \{s_i\} \in X_i$ and $T \leq \{s_{i_1}, s_{i_k}\} \in X_i$. Consequently:
	$$\{s_i\} = \min_{T\in X_i} T,\quad \left|\max_{T\in X_i} T\right| \leq 2\quad \forall i\in\{1,\dots,n\}.$$
	In fact, we have that
	\begin{align*}
		\left| \emptyset\ \triangle\ \min_{T\in X_1} T \right| \leq 0 + 1 < 3, \text{ and }
		\left| \max_{T\in X_i}T\ \triangle\ \min_{T\in X_{i+1}}T \right| \leq 2 + 1 = 3\quad\forall i\in\{1,\dots,n-1\}.
	\end{align*}
\end{proof}\fi

\subsection{Limiting memory space}

Next we examine the memory usage of Algorithm \ref{algo:enumerating_satteams}. Throughout the execution, $\mathcal{D}_2[\{\vec{0}\}]$, $\mathcal{D}_k$ and $\setsatteamsz{k}$ have to be saved. However the size of those lists increases exponentially when raising the size of the outputted teams or the amount of variables occurring in the formula $\varphi$. In general, Algorithm \ref{algo:enumerating_satteams} requires space $\bigo(2^{2^n})$, and $\bigo(2^n)$ when fixing the parameter $k$. In fact, any algorithm that saves a representative system for the orbits of $\ftwo^n\acts\setsatteams{k}$ cannot perform in polynomial space by the following theorem. For this reason we have to discard the group action of flipping bits when limiting memory space to polynomial sizes.

\begin{theorem}
	\label{theorem:amount_orbits_not_polynomial}
	Let $1\neq k\in\natural$ and $n\in\natural$. We set $\varphi \defeq \dep{x_1,x _2, \dots, x_n}$. Then the amount of orbits of $\ftwo^n\acts\setsatteams{k}$ is not polynomial in $n$.
\end{theorem}\iflong
\begin{proof}
	Note that each orbit of $\ftwo^{n-1}$ on the set of all $k$-teams over $n-1$ variables maps to an orbit of $\ftwo^n\acts\setsatteams{k}$ by extending all assignments of a team so that $x_n$ is assigned to the same value. As we have
	$$f(n)\in n^{\bigo(1)} \Leftrightarrow f(n-1)\in n^{\bigo(1)}$$
	for any function $f\colon\natural\rightarrow\natural$, we may assume that $\varphi$ is equivalent to $1$ with $|\var(\varphi)|=n$.
	
	By the Cauchy-Frobenius lemma (see Proposition \ref{prop:cauchy_frobenius}) the amount of orbits is at least
	$$\frac{|\{T\in\setteam{\var(\varphi)}: |T|=k\}|}{2^n}$$
	when neglecting all summands except the one for $\vec{0}\in\ftwo^n$. That is why the number of orbits in $\setsatteams{k}$ has to be larger than $\binom{2^n}{k} / 2^n$, which already increases exponentially in $n$.
\end{proof}\fi



In the previous sections we had to limit the cardinality of outputted teams for obtaining polynomial delay. As the following theorem shows, this measure is necessary as well when demanding polynomial space.

\begin{theorem}
	Let $\Phi$ be any fragment of team-based propositional logic and $f$ be a function with $f\notin n^{\bigo(1)}$ such that for any $n\in\natural$ there exists a formula $\varphi_n\in\Phi$ in $n$ variables with at least $2^{f(n)}$ satisfying teams. Then it follows that $\enumteampl(\Phi)$ cannot be enumerated in polynomial space.
\end{theorem}\iflong
\begin{proof}
	Any enumeration algorithm enumerating $\sol(\varphi_n)$ has to output $2^{f(n)}$ different teams. The same amount of configurations have to be adopted. In order to distinguish these, the configurations are encoded by at least $f(n)$ bits. However, when considering a RAM performing in polynomial space, the contents of all registers may be encoded by a polynomial amount of bits. For this reason a RAM enumerating $\sol(\varphi_n)$ cannot perform in polynomial space.
\end{proof}\fi

\begin{corollary}
	The problem $\enumteampl(\pdlnovee)$ cannot be enumerated in polynomial space.
\end{corollary}\iflong
\begin{proof}
	Let $n\in\natural$. Set $\varphi_n \defeq \dep{x_1,x _2, \dots, x_n}$. All teams $T$ with $s(x_n) = 0$ for all assignments $s\in T$ satisfy $\varphi_n$. For this reason at least $2^{2^{n-1}}$ satisfying teams exist. Because of $2^{n-1}\notin n^{\bigo(1)}$, the claim follows by the previous theorem.
\end{proof}\fi

We now present an algorithm enumerating $\enumteamplsize(\pdlnovee, f)$ for any $f\in n^{\bigo(1)}$ in polynomial space. Compared to Algorithm \ref{algo:enumerating_satteams}, it saves memory space by recomputing the satisfying teams of lower cardinality instead of storing them in a list. As a downside we have to accept incremental delays.

Then, we define a unary relation $\HasNext$ on $\setassign{\var(\varphi)}$ by
	$s\in\HasNext$ if and only if $\exists t\in\setassign{\var(\varphi)} : s < t.$
	For any $s\in\HasNext$ let $\Next(s)$ be the unambiguous assignment such that $s < \Next(s)$ holds but $s < t < \Next(s)$ does not hold for any assignment $t$. We denote the smallest element in $\setassign{\var(\varphi)}$ by $s_\first$.
	The largest element is denoted by $s_\last$.
	As already mentioned when defining the lexicographical order, we assume that $\HasNext$, $\Next$ and $s_\first$ may be determined in $\bigo(n)$ time.

\iflong
\begin{lemma}
	\label{lemma:mc_pdlnovee_fpt}
	Let $T$ be a team with cardinality $k$. Then $T\models\varphi$ can be checked in $\bigo(k^2 |\varphi|)$ time.
\end{lemma}
\begin{proof}
	Because of the $2$-coherence of $\varphi$ it is enough to check all $2$-subteams of $T$. By the proof of Corollary \ref{coro:algorithm_setsatteamsz_complexity} checking a $2$-team is accomplished in $\bigo(|\varphi|)$ time. As $T$ has $\bigo(k^2)$ $2$-subteams, the claim follows.
\end{proof}\fi

\begin{algorithm}[t]
	\caption{Enumerating satisfying teams in polynomial space, ordered by cardinality}
	\label{algo:enumerating_satteams_space}
	
	\KwIn{A team-based propositional formula $\varphi$ as in Equation ($\star$)}
	\KwOut{All teams $T$ for $\varphi$ with $T\models\varphi$, $1\leq|T|\leq f(|\varphi|)$}
	\BlankLine
	
	\For{$k=1,\dots,f(|\varphi|)$}{
		$T \leftarrow \{s_\first\}$\;
		\While{\True}{
			\lIf{$|T| = k$ \And $T\models\varphi$}{\Output $T$}
			$s \leftarrow \max(T)$\;
			\lIf{$|T|< k$ \And $T\models\varphi$ \And $s\in\HasNext$}{
				$T\leftarrow T \cup \{\Next(s)\}$%
				\label{algline:enumerating_satteams_space_increase}%
			}
			\lElseIf{$s\in\HasNext$}{
				$T \leftarrow T \setminus \{s\} \cup \{\Next(s)\}$%
				\label{algline:enumerating_satteams_space_next}%
			}
			\uElseIf{|T| > 1}{
				\iflong
				$T \leftarrow T \setminus \{s\}$\;%
				\label{algline:enumerating_satteams_space_reduce}
				$s \leftarrow \max(T)$\;
				$T \leftarrow T \setminus \{s\} \cup \{\Next(s)\}$\;
				\fi
			}
			\lElse{
				\Break
			}
		}
	}
\end{algorithm}
\iflong
Let $k\in\natural$. We write $\mathcal{M}_k$ for the set of teams $T$ is assigned to during the $k$-th iteration of the outer loop of Algorithm \ref{algo:enumerating_satteams_space}.

\begin{lemma}
	\label{lemma:enumerating_satteams_space_next}
	Let $S\in\mathcal{M}_k$ be a nonempty set such that $s \defeq \max(S)\in\HasNext$. Then it follows that
	$S\setminus\{s\}\cup\{\Next(s)\}\in\mathcal{M}_k.$
	In particular, we have 
	$S\setminus\{s\}\cup\{t\}\in\mathcal{M}_k$ for all $t\in\setassign{\var(\varphi)}$ with $t\geq s$.
\end{lemma}
\begin{proof}
	We conduct an induction over $k - |S|$.
	
	Induction basis ($|S|=k$): Because of $|S|\nless k$ and $s\in\HasNext$ line $\ref{algline:enumerating_satteams_space_next}$ is executed and $T$ is assigned to $S\setminus\{s\}\cup\{\Next(s)\}$.
	
	Induction step ($|S| + 1 \rightarrow |S|$, $1\leq |S| < k$): If $S\not\models \varphi$, line $\ref{algline:enumerating_satteams_space_next}$ is executed as before and the claim follows. If $S\models \varphi$, line \ref{algline:enumerating_satteams_space_increase} is executed. It follows that $S\cup\{\Next(s)\}\in\mathcal{M}_k$. By induction hypothesis it follows that $S\cup \{s_\last\}\in\mathcal{M}_k$. When executing the body of the while loop with $T$ assigned to $S\cup \{s_\last\}$, the block beginning at line \ref{algline:enumerating_satteams_space_reduce} is executed, assigning $T$ to $S\setminus\{s\}\cup\{\Next(s)\}$.
\end{proof}\fi

\iflong
\begin{lemma}
	\label{lemma:enumerating_satteams_space_output}
	Let $k\in\natural$ and $S$ be a team with $S\models\varphi$ and $|S|\leq k$. Then it follows that $S\in\mathcal{M}_k$.
\end{lemma}
\begin{proof}
	We conduct an induction over $|S|$.
	
	Induction basis ($|S|=1$): Clearly $\{s_\first\}\in\mathcal{M}_k$. By Lemma \ref{lemma:enumerating_satteams_space_next} every $1$-team is contained in $\mathcal{M}_k$.
		
	Induction step ($|S| - 1 \rightarrow |S|$, $1<|S|\leq k$): Let $s=\max(S)$. Since $\varphi$ is downward closed, it follows that $S\setminus\{s\}\models\varphi$. The induction hypothesis yields $S\setminus\{s\}\in\mathcal{M}_k$. Consequently the while loop is executed with $T$ assigned to $S\setminus\{s\}$. Line \ref{algline:enumerating_satteams_space_next} is executed, assigning $T$ to a team $S\setminus\{s\}\cup\{t\}$, where $t$ is an appropriate assignment with $t\leq s$. Lemma \ref{lemma:enumerating_satteams_space_next} yields $S\in\mathcal{M}_k$.
\end{proof}\fi
\iflong
\begin{lemma}
	\label{lemma:enumerating_satteams_space_delay}
	Let $f\in n^{\bigo(1)}$ be a polynomial time computable function. Then there exists a polynomial $p$ such that the $i$-th delay of Algorithm \ref{algo:enumerating_satteams_space} is bounded by $i^2p(|\varphi|)$.
\end{lemma}
\begin{proof}
	Note that the delay is constant when outputting the $2^n$ singletons that satisfy $\varphi$ trivially. Hence we assume that $i\geq 2^n$. It is easy to verify that any team $T$ is assigned to is lexicographically larger than the previous value for $T$. For this reason the number of iterations of the inner while loop is bounded by $|\mathcal{M}_k|$.
	
	As $T$ is not assigned to teams with greater cardinality when the current value for $T$ does not satisfy $\varphi$, it follows that $S\setminus\{\max(S)\}\models\varphi$ for any $S\in\mathcal{M}_k$ with $|S|>1$. Consequently 
	$$|\mathcal{M}_k|\leq 2^n\sum_{l=0}^{k-1}\cardsatteams{l}\leq i\sum_{l=0}^{k-1}\cardsatteams{l}.$$
	
	Let $S$ be the ($i+1$)-th outputted element. Set $k=|S|$. We have $S\in\mathcal{M}_k$ and $|S|>1$. By outputting teams of lower cardinality first we guarantee that $i\geq \sum_{l=1}^{k-1}\cardsatteams{l}$. Furthermore $S$ is outputted in the $k$-th iteration of the outer loop. Consequently the inner while loop has been executed at most $k\cdot i^2$ times before outputting $S$. Since $k$ is bounded by a polynomial in $|\varphi|$, by Lemma \ref{lemma:mc_pdlnovee_fpt} it follows that the body of the inner while loop can be executed in polynomial time. We conclude that the $i$-th delay is bounded by $i^2p(|\varphi|)$, where $p$ is an appropriate polynomial.
	
	Now let $i$ be the total amount of outputted teams. Then the number of iterations of the inner while loop is bounded by
	\begin{align*}
	\sum_{k=1}^{f(|\varphi|)} |\mathcal{M}_k| 
	\leq f(|\varphi|) \cdot |\mathcal{M}_{f(|\varphi|)}| 
	\leq f(|\varphi|) \cdot 2^n\sum_{l=0}^{f(|\varphi|)}\cardsatteams{l} 
	\leq f(|\varphi|) \cdot i^2.
	\end{align*}
	Accordingly, we can choose $p$ such that even the postcomputation phase is bounded by $i^2p(|\varphi|)$.
\end{proof}\fi

\begin{theorem}
	Let $f\in n^{\bigo(1)}$ be a polynomial time computable function. Then Algorithm \ref{algo:enumerating_satteams_space} is an \incP-algorithm for $\enumteamplsize(\pdlnovee, f)$ which performs in polynomial space.
\end{theorem}\iflong
\begin{proof}
	The algorithm saves only one team of cardinality $\leq f(|\varphi|)$ and one assignment for which $(f(|\varphi|)+1)$ registers are required. By Lemma \ref{lemma:enumerating_satteams_space_output} it is clear that the algorithm outputs the satisfying teams ordered by cardinality. Lemma \ref{lemma:enumerating_satteams_space_delay} states that the delays conform to the definition of \incP{}.
\end{proof}\fi

\section{Conclusion}
In this paper we have shown that the task of enumerating all satisfying teams of a given propositional dependence logic formula without split junction is a hard task when sorting the output by its cardinality, i.e., only for polynomially sized teams, we constructed a $\delayP$ algorithm.
In the unrestricted cases, we showed that the problem is in $\delayFPT$ when the parameter is chosen to be the team size.
Further, we explained that the algorithm is optimal regarding its output size and pointed out that any algorithm saving a representative system for the orbits of $\ftwo^n\acts\setsatteams{k}$ cannot perform in polynomial space.

Furthermore, we want to point out that allowing for split junction (and accordingly talking about full $\mathcal{PD\!L}$) will not yield any $\delayFPT$ or $\delayP$ algorithms in our setting unless $\p=\np$.

Lastly, we would like to mention that the algorithms enumerating orbits and the satisfying teams, respectively, can be modified such that satisfying teams for formulas of the form $\varphi_1 \classvee \varphi_2 \classvee \cdots \classvee \varphi_r$ with $r\in\natural,\ \varphi_i\in\pdlnovee$ can be enumerated. 
The idea is to merge the outputs $\sol(\varphi_i)$, $i\in\{1,\dots,r\}$, which is possible in polynomial delay if the output for each $\varphi_i$ is pre-sorted according to a total order.

By now, we presented an algorithm that sorts the output by cardinality. 
It remains open to identify the enumeration complexity of Poor Man's Propositional Dependence Logic when other orders, e.g., the lexicographical order, are considered. 
Besides, one can investigate the conjunction free fragment of $\mathcal{PD\!L}$, permitting the split junction operator but no conjunction operator. 
Similarly to the Poor Man's fragment, one can assume that the group action of flipping bits is an invariant for satisfying teams when formulas are simplified properly. 
Nonetheless, the $2$-coherence property is lost so that the algorithm for constructing the sets $\setsatteamsz{k}$ fails.

Finally, we want to close with some questions. 
Are there exact connections or translations to concrete fragments of SQL or relational algebra (relational calculus)? 
Currently, propositional dependence logic can be understood as relational algebra on a finite (and two valued) domain. 
Do the presented enumeration algorithms mirror or even improve known algorithmic tasks in database theory?
Are there better fragments or extensions of $\pdlnovee$ with a broader significance for practice?




\bibliography{enumpdl}

\end{document}